\documentclass[12pt,letter]{article}
\usepackage[left=1in,right=1in,top=1in,bottom=1in]{geometry}
\usepackage{xspace}
\usepackage{amsmath,amsthm, amssymb}
\usepackage{hyperref}
\usepackage{cleveref}

\theoremstyle{plain}
\newtheorem{theorem}{Theorem}[section]
\newtheorem{corollary}[theorem]{Corollary}
\newtheorem{lemma}[theorem]{Lemma}
\newtheorem*{mainpoly}{Theorem \hyperref[thm:mainpoly]{A}(=\Cref{thm:mainpoly})}
\newtheorem*{mainbool}{Theorem \hyperref[thm:main]{B} (=\Cref{thm:main})}

\newtheorem*{claim}{Claim}
 
\theoremstyle{definition} 
\newtheorem{definition}[theorem]{Definition}
\newtheorem{open}[theorem]{Open Question}

\newtheorem{remark}[theorem]{Remark}
\newtheorem{observation}[theorem]{Observation}

\newcommand{\cc}[1]{\mathsf{#1}}
\newcommand{\F}{\mathbb{F}}
\newcommand{\K}{\mathbb{K}}
\newcommand{\Z}{\mathbb{Z}}
\newcommand{\Q}{\mathbb{Q}}
\newcommand{\poly}{\text{poly}}

\newcommand{\lang}[1]{\textsc{#1}\xspace}

\newcommand{\UNSAT}{\lang{UNSAT}}
\newcommand{\PIT}[1][]{\lang{PIT\textsubscript{#1}}}
\DeclareMathOperator{\ch}{char}

\DeclareMathOperator{\BIT}{BIT}
\DeclareMathOperator{\ADD}{ADD}

\DeclareMathOperator{\SUBPROD}{SUBPROD}
\DeclareMathOperator{\PROD}{PROD}

\DeclareMathOperator{\LT}{LT}

\DeclareMathOperator{\VAL}{VAL}
\DeclareMathOperator{\VEC}{VEC}
\DeclareMathOperator{\REM}{REM}
\DeclareMathOperator{\INV}{INV}

\newlength{\mylength}
\newcommand{\algprob}[4]{
\begin{quotation}
\begin{tabbing}
\hphantom{Decide: } \= {#1} \\
\textit{Input: } \> {#2} \\
\textit{Decide: } \> {#3} \\
 \> \parbox{\dimexpr\linewidth-\the\mylength}{#4}
\end{tabbing}
\end{quotation}
}

\DeclareMathOperator{\gates}{gates}
\DeclareMathOperator{\wires}{wires}
\DeclareMathOperator{\depth}{depth}
\DeclareMathOperator{\proddepth}{prod-depth}

\newcommand{\Hrubes}{Hrube\v{s}\xspace}

\title{Polynomial Identity Testing and the Ideal Proof System: \\
\PIT is in $\cc{NP}$ if and only if IPS can be p-simulated by a Cook--Reckhow proof system}
\author{Joshua A. Grochow}

\begin{document}
\maketitle

\begin{abstract}
The Ideal Proof System (IPS) of Grochow \& Pitassi (FOCS 2014, \emph{J. ACM}, 2018) is an algebraic proof system that uses algebraic circuits to refute the solvability of unsatisfiable systems of polynomial equations. One potential drawback of IPS is that verifying an IPS proof is only known to be doable using Polynomial Identity Testing (PIT), which is solvable by a randomized algorithm, but whose derandomization, even into $\cc{NSUBEXP}$, is equivalent to strong lower bounds. However, the circuits that are used in IPS proofs are not arbitrary, and it is conceivable that one could get around general PIT by leveraging some structure in these circuits. This proposal may be even more tempting when IPS is used as a proof system for Boolean Unsatisfiability, where the equations themselves have additional structure.

Our main result is that, on the contrary, one cannot get around PIT as above: we show that IPS, even as a proof system for Boolean Unsatisfiability, can be p-simulated by a deterministically verifiable (Cook--Reckhow) proof system if and only if PIT is in $\cc{NP}$. We use our main result to propose a potentially new approach to derandomizing PIT into $\cc{NP}$.

\end{abstract}

\section{Introduction} \label{sec:intro}
Proof complexity studies the lengths of proofs in various formal proof systems, and sits at the intersection of algorithms, logic, and computational complexity. It is often the case that we can view the run of an algorithm as a proof of its output; if we can then formalize a proof system in which such proofs live, then lower bounds on the size of proofs in that proof system imply lower bounds on the complexity of the algorithm (runtime, space, etc.). 

One of the powerful features of proof complexity is that, once such a proof system is identified, it is often the case that the proof system captures not just one algorithm, but a \emph{whole family} of algorithms. For example:
\begin{itemize}
\item the Resolution proof system \cite{robinson} for Boolean formulas captures all algorithms for Boolean Satisfiability in the DPLL family; 

\item the Cutting Planes proof system \cite{CCT} (originating in the cutting planes method by Gomory \cite{gomory1963} and Chv\'{a}tal \cite{chvatal1973}) captures a family of branch-and-bound algorithms based on certain tactics in Integer Linear Programming; 

\item and the Sum-of-Squares (or Positivstellensatz) proof system \cite{GV} captures algorithms based on Semi-Definite Programming relaxations of integer optimization problems. 
\end{itemize}
An example most relevant for us is:
\begin{itemize}
\item the Polynomial Calculus proof system \cite{CEI} captures algorithms that are based on manipulation of polynomials, such as Gröbner basis algorithms.
\end{itemize}

It is natural to ask how such algebraic proof systems relate to more standard Boolean proof systems. Classical line-by-line derivations of Boolean tautologies based on the \emph{modus ponens} rule (from $A$ and $A \Rightarrow B$, derive $B$) form the Frege family of proof systems. If we restrict each line to obey the syntactic restriction associated to some circuit class $\mathcal{C}$, we get so-called $\mathcal{C}$-Frege; for example, when $\mathcal{C}=\cc{AC}^0$, each line must be expressed as a formula of constant depth. When $\mathcal{C}=\cc{AC}^0[p]$, the corresponding proof system is closely related to algebraic proof systems. Indeed, an additional impetus to the development of such systems was to get $\cc{AC}^0[p]$-Frege lower bounds, a problem which not only remains open to this day (unlike $\cc{AC}^0[p]$ circuit lower bounds \cite{Razborov, Smolensky}), but for which we have little formal evidence that its resolution should be difficult. When $\mathcal{C}=\cc{P/poly}$, we get the Extended Frege proof system, which can work with each line of the proof expressed as an arbitrary Boolean circuit. Among the well-studied proof systems for Boolean tautologies, Extended Frege is perhaps the most powerful. 

In part to seek a new route to proof complexity lower bounds, in 2014 Grochow and Pitassi introduced the Ideal Proof System (IPS) \cite{GrochowPitassi}, which is at least as strong as Extended Frege.

\begin{definition}[{Ideal Proof System, Grochow \& Pitassi \cite{GrochowPitassi} (cf. \cite{Pitassi96, Pitassi98})}]
Let $R$ be a ring. An \emph{IPS$_R$ certificate} that a polynomial $G(\overline{x})$ is in the ideal generated by polynomials $F_1(\overline{x}), \dotsc, F_m(\overline{x})$ is a polynomial $C(\overline{x}, \overline{y})$ over $R$ such that
\begin{enumerate}
\item $C(\overline{x}, \overline{F}(\overline{x})) = G(\overline{x})$, and
\item $C(\overline{x}, \overline{y})$ is in the ideal $\langle y_1, \dotsc, y_m \rangle \subseteq R[\overline{x}, \overline{y}]$; equivalently, $C(\overline{x}, \overline{0})=0$.
\end{enumerate}
An \emph{IPS proof} or \emph{IPS derivation} of $G$ from $\{F_1, \dotsc, F_m\}$ is an $R$-algebraic circuit computing an IPS certificate, and an \emph{IPS refutation} (of the solvability) of $\{F_1, \dotsc, F_m\}$ is an IPS proof of the constant polynomial $1$ from $\{F_1, \dotsc, F_m\}$. (We sometimes omit the subscript $_R$ when it is clear from context.)
\end{definition}

\newcommand{\alg}{\text{alg}}
IPS can be used as a proof system for unsatisfiable Boolean CNFs (the \lang{UNSAT} problem) as follows. Boolean formulas are translated into systems of polynomial equations as follows:
\[
\alg(x) := x \qquad \alg(\neg x) := 1-x \qquad \alg(\varphi \wedge \psi) := \alg(\varphi) \cdot \alg(\psi)
\]
and then to claim that a Boolean formula $\varphi$ is satisfied, we include the polynomial $1-\alg(\varphi)$ in our initial set of polynomials. The remaining initial polynomials are the ``Boolean axioms'' $x_i^2 - x_i$, for each variable $x_i$, as setting these to zero enforces that in any potential solution, each $x_i$ must be $\{0,1\}$-valued. When comparing IPS to Boolean proof systems such as $\cc{AC}^0[p]$-Frege or Extended Frege, it is typically this particular application of IPS as a proof system for \lang{UNSAT} that is used.

A key conceptual advantage of IPS is that proofs are \emph{just algebraic circuits}, so that one might hope to use algebraic circuit lower bounds techniques to get lower bounds on proof systems. Indeed, this was achieved for several restricted versions of IPS by Forbes, Shpilka, Tzameret, and Wigderson \cite{FSTW}, Andrews \& Forbes \cite{AndrewsForbes}, and Govindasamy, Hakoniemi, and Tzameret \cite{GHT}. These lower bounds were not for IPS as a proof system for \UNSAT, but rather on restricted versions of IPS as a proof system for:

\algprob{\lang{Variety Emptiness}\footnote{In the literature sometimes referred to as the computational problem \lang{Hilbert's Nullstellensatz}; we prefer the name \lang{Variety Emptiness} because the name is more descriptive, and less prone to confusion with the theorem by the name of Hilbert's Nullstellensatz.} over a field $\F$}
{A set $\mathcal{F}$ of polynomials over $\F$}
{Is $\{F=0 : F \in \mathcal{F}\}$ unsolvable over the algebraic closure $\overline{\F}$?}
{Equivalently, is the variety $Z(\mathcal{F})$ empty?}

However, one drawback of IPS is that, precisely because of its use of algebraic circuits, proofs are \emph{not} known to be verifiable deterministically in polynomial time, in contrast to the other line-by-line proof systems mentioned above. Over fields, IPS proofs can nonetheless be verified in randomized polynomial time \cite{Pitassi96, GrochowPitassi}, using Polynomial Identity Testing---known to be in $\cc{coRP}$ \cite{DML, Schwartz, Zippel}---on the IPS derivation:

\algprob{\lang{Polynomial Identity Testing (\PIT)}}
{A field $\F$ and an algebraic circuit $C$ over $\F$}
{Does $C$ compute the identically zero polynomial?}
{Equivalently, is $C(\overline{\alpha})=0$ for all $\overline{\alpha} \in \overline{\F}^n$?}

Some restricted subsystems of IPS can be verified by restricted versions of \PIT that are known to be in $\cc{P}$, whereby those subsystems of IPS are deterministically verifiable. This, along with showing that the Raz--Shpilka \PIT algorithm for non-commutative formulas \cite{RazShpilka} could be formalized in Frege, allowed Li, Tzameret, and Wang to show that non-commutative formula IPS was quasi-polynomially equivalent to Frege \cite{LTW}. But in general we do not know whether the use of \PIT for verifying IPS proofs can be circumvented.

This paper is about precisely this issue. Derandomizing \PIT---even getting it into $\cc{NSUBEXP}$---is essentially equivalent to certain flagship lower bounds \cite{KI, CIKK, JansenSanthanam}. Thus, while it is widely believed that $\PIT$ \emph{can} be derandomized (and, consequently, that IPS would be deterministically verifiable), doing so involves resolving a major open lower bounds question.

Here we come to the first question answered by our main results. It seems intuitively clear that one cannot verify an IPS certificate without using \PIT, but 
\begin{quotation}
\noindent \textbf{Question 1:} Is it possible that IPS certificates are ``special'' in some way that makes \PIT for those circuits easier?
\end{quotation}
Indeed, already in \cite[Sec.~6]{GrochowPitassi} it was observed that the IPS derivations for given polynomials $G$ and $F_1, \dotsc, F_m$ form a coset of an ideal, and for IPS refutations (when $G=1$), that coset is even closed under multiplication. This is quite a lot of structure! Is there some additional structure to IPS certificates that makes \PIT easier? Our main result implies an essentially negative answer to Question 1. (We discuss the ``essentially'' in Sections~\ref{sec:results} and \ref{sec:open}.)

\newcommand{\pflength}{\text{proof-len}}
The second question answered by our main results is about the relationship between IPS and deterministically verifiable (a.k.a. Cook--Reckhow) proof systems. To state this question we need the notion of p-simulation. Given a proof system $\mathcal{P}$ for a language $L$, let $\pflength_{\mathcal{P}}(x)$ denote the length of the shortest $\mathcal{P}$-proof that $x \in L$ (we take this to be infinite for $x \notin L$). Then we say $\mathcal{P}$ \emph{p-simulates} another proof system $\mathcal{Q}$ for $L$ if, for every $x \in L$, 
\[
\pflength_{\mathcal{P}}(x) \leq \poly(\pflength_{\mathcal{Q}}(x)).
\]
That is, if what we care about is the distinction between polynomial-length versus super-polynomial-length proofs, then $\mathcal{P}$ is at least as powerful as $\mathcal{Q}$. 

In an early version of \cite{GrochowPitassi}, we had said (paraphrasing) ``Surely nothing in our paper has gone beyond ZFC, so IPS is p-simulated by ZFC.'' It was pointed out to us by Pavel \Hrubes and an anonymous reviewer that this was not so obvious. That feedback was in fact the genesis of this paper (see \hyperref[sec:origin]{Origin of the paper}, below), and brings us to the second question, which we answer (in the negative):
\begin{quotation}
\noindent \textbf{Question 2:} Is it possible for IPS to be p-simulated by \emph{some} Cook--Reckhow proof system, without derandomizing \PIT?
\end{quotation}

For IPS derivations in general, Alekseev, Grigoriev, Hirsch, and Tzameret \cite[Footnote~5]{AGHT} observed that \PIT can be solved by the \lang{IPS Verification} problem:

\algprob{\lang{IPS Verification}}
{An algebraic circuit $C$, and polynomials $G, F_1, \dotsc, F_m$}
{Is $C$ an IPS certificate that $G \in \langle F_1, \dotsc, F_m \rangle$?}{}
Their observation is that a polynomial $G$ computed by a circuit $C$ is derivable from the zero polynomial, if and only if $G$ is itself the zero polynomial, that is, iff $C \in \PIT$. Put another way, IPS certificates that derive 0 from 0 are precisely the same thing as circuits that compute the zero polynomial, and thus $\PIT \leq_m^p \lang{IPS Verification}$. Since $\lang{IPS Verification}$ can be solved by $\PIT$, we have $\lang{IPS Verification} \equiv_m^p \PIT$ (many-one,\footnote{The straightforward use of \PIT to verify an IPS proof is to query $C(\overline{x}, \overline{0})$ and $G(\overline{x})-C(\overline{x}, \overline{F}(\overline{x}))$. To get this to be many-one, we instead use the AND-function for \PIT: introduce two new variables $z_1, z_2$ and ask whether $z_1 C(\overline{x}, \overline{0}) + z_2(G(\overline{x})-C(\overline{x}, \overline{F}(\overline{x})))$ is identically zero.} polynomial-time equivalence), so one can be derandomized iff the other one can.

However, this observation does not seem to tell us much about \lang{IPS Verification for Variety Emptiness}, namely, instances of \lang{IPS Verfication} where $G=1$ (and the system of equations $F_1=\dotsb=F_m=0$ is unsatisfiable over the algebraic closure $\overline{\F}$), let alone about whether IPS can be p-simulated by a Cook--Reckhow proof system. Indeed, the observation above is about using IPS as a proof system for the language ``Can 0 be derived from 0?'', which has a trivial and efficient Cook--Reckhow proof system, even though IPS for that language is not Cook--Reckhow unless $\PIT \in \cc{P}$. 

Further, from the complexity point of view, showing that a randomized proof system $\mathcal{P}$ is p-simulated by a Cook--Reckhow system is nearly as good as showing that $\mathcal{P}$ itself is Cook--Reckhow. For example, if $\mathcal{P}$ is a randomized proof system for a $\cc{coNP}$-complete language, and $\mathcal{P}$ is p-simulated by a Cook--Reckhow system, then $\mathcal{P}$ being p-bounded still implies $\cc{NP} = \cc{coNP}$.

In the case of \lang{Variety Emptiness} (a $\cc{coNP}$-hard language that is in $\cc{PSPACE}$ in general \cite{ierardi}, and in $\cc{AM}$ in characteristic zero assuming the Generalized Riemann Hypothesis \cite{koiran}), one may wonder whether deriving $1$ from a system of polynomial equations forces enough structure on an IPS certificate to avoid needing the full strength of \PIT for verification. In the case of Boolean \UNSAT, one may wonder further if the additional structure coming from the fact that the polynomials $F_i$ are the degree-3 translations of width-3 CNF clauses, together with the Boolean axioms, is enough to avoid the necessity of \PIT. 

\subsection{Main results, and an approach to derandomzing \PIT} \label{sec:results}
Our main results are that, in both of the above settings, not only does p-simulating IPS by a Cook--Reckhow system imply that \PIT is in $\cc{NP}$, but in fact the two are equivalent. 

\begin{mainpoly} 
\PIT[$\F$] is in $\cc{NP}$ if and only there exists a Cook--Reckhow proof system that p-simulates IPS$_\F$ for \lang{Variety Emptiness} over $\F$.
\end{mainpoly}

The preceding result in fact works for arbitrary fields $\F$ if we replace $\cc{NP}$ by $\cc{NP}_{\F}$ in the Blum--Shub--Smale model \cite{BSS} over $\F$. However, when $\F$ is a finite field we have $\cc{NP}_{\F} \equiv_m^p \cc{NP}$, and when $\F$ is an algebraic number field and we measure IPS size by total bit-size, the result still holds with the usual Boolean definition of $\cc{NP}$. In the following result, we restrict our fields to finite fields or the rationals, and measure IPS size by total bit-size.

\begin{mainbool}
If there is a Cook--Reckhow proof system that p-simulates IPS for Boolean \UNSAT over fields of size $\leq \poly(q, 2^n)$ (where $n$ is the size of the CNF), then \PIT is in $\mathsf{NP}$ for circuits over $\Q$ or over finite fields of size at most $q$.

Conversely, for any field $\F$ if \PIT[$\F$] is in $\mathsf{NP}$, then there is a Cook--Reckhow proof system that p-simulates IPS$_\F$ for \UNSAT.
\end{mainbool}

Unlike the case of IPS for \lang{Variety Emptiness}, in our result for IPS for \UNSAT, there is a gap in the forward direction between the field IPS is over and the field \PIT is over. We discuss the source of this gap and the difficulty of closing it in \Cref{rmk:field} and \Cref{sec:open}.

The forward direction of our results, in combination with results of Grochow \& Pitassi \cite{GrochowPitassi}, suggest a potentially new approach for showing that \PIT is in $\cc{NP}$. Namely, for any sufficiently powerful proof system $\mathcal{P}$ (say, above $\cc{AC}^0$-Frege) they showed  that if there is a family of Boolean circuits $K$ solving \PIT[$\F$], and such that the so-called ``PIT axioms for $K$'' (which they introduce, and we recall in \Cref{sec:pit}) have short $\mathcal{P}$-proofs, then IPS$_\F$ is p-simulated by $\mathcal{P}$. If $\mathcal{P}$ is Cook--Reckhow, then by our results this would put \PIT into $\cc{NP}$. 

One feature we find interesting about this proposal is that, unlike proposals based on uniform pseudo-random generators, or uniform algorithms for special cases of \PIT, our proposal is almost entirely non-uniform. All the uniformity has been ``swept under the rug'' into the verifier for $\mathcal{P}$. If one can merely show that such circuits $K$ \emph{exist}, and such $\mathcal{P}$-proofs of the PIT axioms for $K$ \emph{exist}, then it implies the uniform conclusion that \PIT is in $\cc{NP}$.

Without the details of the PIT axioms, one might think this is a triviality. Namely, if there were polynomial-size circuits $K$ for \PIT, and polynomial-size proofs that those circuits were correct, then to solve \PIT in $\cc{NP}$, the nondeterministic machine can guess both the circuit $K$ and its proof of correctness, and then use $K$ to solve \PIT (cf. $\cc{CertP/poly}$ \cite{certppoly}). A crucial aspect of this proposal, therefore, is how relatively \emph{weak} the PIT axioms are compared to fully proving that a circuit $K$ correctly solves \PIT. 

We will discuss this in more detail in \Cref{sec:pit}, but for now we highlight one important aspect. Three of the four PIT axioms feel routine, and we expect would be easy to prove for circuits $K$ solving \PIT. The trickiest, and likely strongest, of the PIT axioms is the implication that if $K$ says an algebraic circuit $C$ is identically zero, then it should also say that $C$, when Boolean variables $\overline{p}$ are plugged in for its algebraic variables, is still zero, that is, that $C$ evaluates to zero on the entire Boolean cube. Mathematically this is a triviality, but the point is that the proof system cannot enumerate over the Boolean cube to prove it, as such a proof would necessarily have exponential size. Instead the proof must work on the resulting tautology where $\overline{p}$  are left as variables.

\subsection{Proof technique for the main results}
The converse directions of both results go back to \cite{Pitassi96, GrochowPitassi}, as they follow just because IPS$_\F$ verification can be done with \PIT[$\F$]. Here we outline our proofs for the forward directions.

\begin{proof}[Proof outline for \Cref{thm:mainpoly}]
For the forward direction, we start with an algebraic circuit $C$ over a field $\F$, and want to test (nondeterministically) whether $C$ is the identically zero polynomial. For IPS for \lang{Variety Emptiness}, the idea is to encode $C$ into a system of polynomial equations gate by gate, e.g. for the multiplication gate $v=u \times w$, we use the equation:
\[
g_v - g_u \cdot g_w=0
\]
where $g_v$, etc., are new variables. We then add one new variable $z$ and the additional equation 
\[
1-zC=0,
\]
whose solutions enforce that the output of $C$ is invertible (\Cref{lem:algebraic}, which actually works over an arbitrary ring). 

The key is to show that IPS has a short refutation of this system of equations. If $f_v$ is the polynomial computed at the gate $v$, we show by induction on the structure of $C$ that IPS can derive $g_v - f_v$ for every gate $v$. If IPS for \lang{Variety Emptiness} is p-simulated by a Cook--Reckhow system $\mathcal{P}$, then the nondeterministic algorithm is to deterministically produce the preceding equations, and then guess the short $\mathcal{P}$-refutation of them and verify it. 
\end{proof}

In addition to proving our first result, \Cref{lem:algebraic} will also play a role in the proof for IPS for Boolean \UNSAT.

\begin{proof}[Proof outline for \Cref{thm:main}]
For IPS for Boolean \UNSAT, we ultimately need to encode the preceding system of polynomial equations into a Boolean CNF, that IPS can still efficiently reason about. Our strategy is as follows. First (\Cref{lem:field}), we find a field $\K$ such that $C$ computes a nonzero \emph{function} over $\K$ iff $C$ was not the identically zero polynomial. If the original field $\F$ was a finite field, then we may take $\K$ to be an extension field whose degree is sufficiently large (larger than $\log_2$ of the syntactic degree of $C$ suffices, so, large enough, but still polynomial in the input size). When the original field $\F$ is the rationals, we may take $\K$ to be $\Z/p\Z$ for a sufficiently large prime $p$---again, we find that there exists a $p$ of polynomial bit-size that suffices. In either case, $\K$ is a finite field.

Then we build the equations above, as in the case of \lang{Variety Emptiness}, but now thought of as algebraic circuits over the field $\K$. We then encode these equations as Boolean circuits. Moving from $\K$ to the prime field $\F_p \subseteq \K$ uses the standard trick of viewing $\K$ as a vector space over $\F_p$; \Cref{lem:extension} ensures that IPS$_\K$ can efficiently recover the original equations over $\K$ from the new equations over $\F_p$. (Note that if $[K : \F_p] = e$, then there are $e$ times as many equations over $\F_p$ than over $\K$, in $e$ times as many variables.)  

To go from these algebraic circuits over prime fields to Boolean CNFs, we now encode arithmetic over $\Z/p\Z$ in a Boolean circuit. The key here, again, is that from one of the Boolean circuits we build, IPS$_{\F_p}$ can efficiently reconstruct the original polynomial over $\F_p$. Here we must work in a binary encoding, not unary as in Santhanam \& Tzameret \cite{SanthanamTzameret}, as $|\F_p|$ could be exponential in $n$ (e.g., in the case we started over $\Q$). Much of the work to encode binary arithmetic into Boolean circuits that can be reasoned about by IPS was done in \cite{AGHT}, for circuits over $\Z$. Our main contribution here (\Cref{lem:rem}) is to extend this with a remainder-modulo-$p$ Boolean circuit $\REM_p$, such that IPS (in characteristic $p$) can efficiently prove that the input and output of $\REM_p$ are two bit-strings that encode equal elements of $\F_p$. This then lets us prove \Cref{lem:binvalp}, which is a mod-$p$ version of Alekseev \emph{et al.}'s binary value principle. 

The unsatisfiable CNF we produce from $C$ is thus the end of all these procedures: 
\begin{enumerate}
\item Find $\K$, 
\item build the equations as above over $\K$, 
\item encode them as more equations in more variables over the prime field $\F_p \subseteq \K$, and  
\item then encode \emph{those} equations using Boolean circuits, and 
\item finally use the usual translation from circuits to CNFs. 
\end{enumerate} 
Now, suppose IPS$_\K$ is p-simulated by a Cook--Reckhow system $\mathcal{P}$. Then the $\cc{NP}$ algorithm is to guess the $\mathcal{P}$-refutation of the resulting CNF (which is unsatisfiable iff $C$ computes the identically zero polynomial). 

The key is to show that IPS, and hence $\mathcal{P}$, has polynomial-size refutations of this CNF. And for this, we basically read the above paragraphs in reverse order. From the CNF, IPS (over any ring $R$) can reconstruct the Boolean circuits (\Cref{lem:circuit-sat}). From the Boolean circuits, the mod-$p$ binary value principle (\Cref{lem:binvalp}) tells us that IPS over any field of characteristic $p$ can efficiently reconstruct the equations over $\F_p$. \Cref{lem:extension} then ensures that IPS$_\K$ can efficiently reconstruct the equations over $\K$, and finally \Cref{lem:algebraic} says that IPS$_\K$ has a short refutations of the equations over $\K$. This completes the outline of the proof.
\end{proof}

\begin{remark} \label{rmk:field}
It is the use of the field $\K$ that is responsible for the size bounds of $2^{\poly(n)}$ in the statement of \Cref{thm:main}, and why it is only ``essentially'' an equivalence, since the field $\F$ over which we put \PIT into $\cc{NP}$ is not always the same as the field $\K$ for which IPS$_\K$ for Boolean \UNSAT is p-simulated by a Cook--Reckhow system. Removing this ``essentially'' is an interesting question. Here we just highlight why it seems to be needed in our present proof. 

When $\F$ is a finite field, if we had not used the field $\K$, then it is possible that $C$ is the zero function over $\F$, while not being the zero polynomial, e.g. the polynomial $x^3 - x$ over $\F=\F_3$. In this case, it is possible that the equations that end with $1-zC=0$ are in fact satisfiable over an extension field of $\F$, even though they are unsatisfiable over $\F$ itself. When we translate to Boolean equations, because of the Boolean axioms $x_i^2 -x_i$, the Boolean equations only ``see'' the field $\F_p \subseteq \K$, so they would report that this system of equations was unsatisfiable, when in fact from the algebraic perspective they should be satisfiable, thus giving the wrong answer for \PIT.

When $\F=\Q$ is the rationals, a similar issue arises. Namely, to encode $C$ into Boolean circuits, we must limit the bit-size of the rationals that we consider. If we consider rationals of too small a bit-size, it is possible that $C$ is nonzero (as both a polynomial and a function over $\Q$), but that $C$ evaluates to zero on all inputs of small bit-size. In this case, again we have that the CNF sees an all-zero function, when we needed it to see that $C$ was nonzero. (Over $\Q$ there is also the issue that we can't choose a small enough bit-size to faithfully represent the function without moving to a finite field. For example, by repeated squaring, the polynomial $x^{2^n}$ has a circuit of size $n$, but even on input $2$, the bit-size needed for the output is $2^n$.)
\end{remark}

\section{Preliminaries}
\subsection{Proof complexity}
Let $\Sigma$ be a finite alphabet, $\Sigma^*$ the set of all finite words over $\Sigma$. A \emph{Cook--Reckhow proof system} for a language $L \subseteq \Sigma^*$ is a polynomial-time function $\mathcal{P}$ whose image is precisely all of $L$. One may think of the inputs to $\mathcal{P}$ as proofs, and the output of $\mathcal{P}$ as the statement proved by the proof. For $x \in L$, if $\mathcal{P}(\pi)=x$, $\pi$ is said to be a $\mathcal{P}$-proof that $x$ is in $L$. 

A \emph{probabilistic proof system} for a language $L$ (see, e.\,g., \cite[Def.~2.6]{GrochowPitassi}) is a probabilistic polynomial-time algorithm $\mathcal{P}$ such that (1) $\mathcal{P}(\pi) \in L$ for all $\pi$, and (2) there is a surjective function $f \colon \Sigma^* \to L$ such that, for all $\pi \in \Sigma^*$, $\mathcal{P}(\pi) = f(\pi)$ with probability at least $2/3$ over $\mathcal{P}$'s random choices.

A proof system $\mathcal{P}$ (Cook--Reckhow or probabilistic) for a language $L$ is \emph{polynomially bounded} or \emph{p-bounded} if there is a polynomial $p$ such that for all $x \in L$, there is a $\mathcal{P}$-proof $\pi$ that $x \in L$ with $|\pi| \leq p(|x|)$. If $L$ has a p-bounded Cook--Reckhow proof system then $L$ is in $\cc{NP}$ \cite{CookReckhow}; if $L$ has a p-bounded probabilistic proof system then $L$ is in $\cc{MA}$ (e.\,g., \cite[Sec.~2]{GrochowPitassi}).

\subsection{Rings, fields, and Polynomial Identity Testing}
By ring we mean a commutative ring with unit; ring homomorphisms must send $1$ to $1$. If $R$ is a ring and $r_1,\dotsc,r_m \in R$, then the ideal they generate is denoted $\langle r_1, \dotsc, r_m \rangle$. When $m=1$, we may write the quotient ring $R / \langle r_1 \rangle$ by $R / r_1$. 

The characteristic of a field $\F$ is the smallest integer $n$ such that $1 + 1 + \dotsc + 1 = 0$ (adding $1$ to itself $n$ times), or $0$ if no such $n$ exists. The characteristic of a field is necessarily always 0 or a prime $p$. The \emph{prime field} of characteristic $p \geq 0$ is the unique smallest field of that characteristic: the rational field $\Q$ is the prime field of characteristic zero, and for $p > 0$ prime, the ring of integers mod $p$, $\Z/p\Z = \F_p$, is the prime field of characteristic $p$.

A field $\K$ is an extension of a field $\F$ of degree $e$ if $\F \subseteq \K$ and the dimension of $\K$ as an $\F$-vector space is $e$. In this case we may write $[\K : \F]=e$ for the degree of the extension. 
Unless otherwise specified, all fields we consider will be finite-degree extensions of the prime field of the same characteristic, that is, either finite fields or algebraic number fields (=finite degree extensions of $\Q$).

When we say a field $\F$ of characteristic $p$ is ``given'' (e.g., as input to a computational problem) we mean that $p$ is specified in binary, and the coefficient vector of an irreducible square-free polynomial $f$ over the prime field of characteristic $p$ such that $\F = \F_p[x] / f(x)$ if $p > 0$ or $\F=\Q[x] / f(x)$ if $p=0$. Once $\F$ is given by such a pair $(p,f)$, elements of $\F$ are given as vectors of dimension $\deg f$ over the prime field of characteristic $p$. In characteristic zero, rationals are specified in the usual way by a pair of integers written in binary, and in characteristic $p$, an element of the prime field $\F_p$ is specified by the usual binary description of one of the integers $\{0, \dotsc, p-1\}$. A polynomial $g$ over $\F$ is given by its list of nonzero coefficients (and the corresponding exponents of the associated monomial) unless otherwise specified.

Finally, we will need the following result.

\begin{lemma}[{Polynomial Identity Testing Lemma, \cite{DML, Schwartz, Zippel}\footnote{The finite field version of this result goes back to Ore \cite{ore}. A beautifully simple proof of the result over finite fields was given by Moshkovitz \cite{moshkovitz}.}}] \label{lem:pit}
Let $f$ be an $n$-variable polynomial over a field $\F$, and $S \subseteq \F$. If $f$ is not the zero polynomial, then
\[
\Pr_{x \in S^n}[f(x) = 0] \leq \frac{\deg f}{|S|}.
\]
\end{lemma}

\subsection{Circuits}
An algebraic circuit $C$ over a ring $R$ is a directed acyclic graph in which each source is labeled by a variable $x_i$ (variables may appear multiple times) or an element of $R$, and each non-source vertex is labeled either as a multiplication gate, an inversion gate (of in-degree 1), or a linear combination gate. Multiplication gates have in-degree at most 2, while linear combination gates may have arbitrary in-degree. We call a gate $v$ \emph{syntactically constant} if every input that has a directed path to $v$ is a constant; in this case, the constant computed at $v$ is independent of the input, and we denote it $f_v$, for consistency with our notation for other gates. Division gates are only allowed when their denominator $v$ is syntactically constant and $f_v$ is invertible in $R$. The incoming edges to a linear combination gate are labeled by elements of $R$. Each gate computes a polynomial over $R$ in the following standard inductive way: input gates compute the polynomial corresponding to their label. For other gates $v$, let $f_v$ denote the polynomial computed at $v$. If $v = u \times w$ is a product gate, then $f_v = f_u \cdot f_w$. If $v$ is a linear combination gate with incoming edges from gates $u_1, \dotsc, u_\ell$ with corresponding constants $c_1, \dotsc, c_\ell$ on the edges, then $f_v = \sum_{i=1}^\ell c_i f_{u_i}$. If $v=1/u$ is an inversion gate, then $f_v = 1/f_u$ (recall such gates are only allowed when $u$ is syntactically constant and a unit in $R$). The polynomial(s) computed by $C$ are the list $(f_{v_1}, \dotsc, f_{v_k})$ where $v_1, \dotsc, v_k$ are the sink gates, also called output gates. 
 
 There are several measures of the ``size'' of a circuit. The \emph{depth} of a circuit $C$, denoted $\depth(C)$, is the longest direct path from any source (input) to any sink (output). We denote the number of edges of $C$ by $\wires(C)$ and the number of gates by $\gates(C)$. When we speak of the ``size'' of a circuit, we may mean wires or gates; up to polynomial factors the choice is immaterial. The \emph{product-depth} of $C$ is the maximum number of multiplication gates on any directed path from an input to an output, which we denote $\proddepth(C)$.
 
A circuit is \emph{constant-free} if the only constants from $\F$ used in the circuit are $\{0,1,-1\}$; other constants may be built up from these using gates. Note that our constant-free circuits still allow division gates by syntactically constant nonzero denominators; this is the same as \cite[Def.~14]{AGHT}. Over the rational numbers, constant-free circuits of polynomial size compute the same class of functions as circuits of polynomial total bit-size. Over any field $\F$, constant-free circuits can only compute polynomials over the prime field contained in $\F$, since they have no way to build constants that are outside of the prime subfield. 

\newcommand{\sdeg}{\text{sdeg}}
\begin{definition}[Syntactic degree]
The \emph{syntactic degree} $\sdeg(C)$ of a circuit $C$ is defined recursively as follows: 
\[
\sdeg(v) = \begin{cases}
1 & \text{ if $v$ is a variable or constant input gate} \\
\max\{\sdeg(v_i) : i = 1, \dotsc, k\} & \text{ if $v = \sum_{i=1}^k \alpha_i v_i$} \\
\sdeg(u) + \sdeg(w) & \text{ if $v = u \times w$} \\
\sdeg(u) & \text{ if $v = 1 / u$}
\end{cases}
\]
The syntactic degree of a circuit $C$ is the maximum syntactic degree of its output gates.
\end{definition}
Note that we have defined the syntactic degree of a constant input to be $1$, not zero as might be expected for a constant; also the syntactic degree of $1/u$ is the same as $u$, \emph{not} the negative of $\sdeg(u)$ as one might expect. These are both in order to account for bit-size, that is, in order to make part 2 of the following observation hold.
The following observation is standard, and can be proved by induction:

\begin{observation} \label{obs:sdeg}
\begin{enumerate}
\item The degree of the polynomial(s) computed by $C$ is always at most $\sdeg(C)$.

\item The bit-length of any coefficient appearing in the polynomial computed at any gate of $C$ is at most $\log_2(\sdeg(C))$. 

\item The syntactic degree of a circuit $C$ with product gates of fan-in $2$ is at most $2^{\proddepth(C)}$.
\end{enumerate}
\end{observation}

One advantage of the syntactic degree over the usual degree is that the syntactic degree is easily calculated (in logspace) from the structure of the circuit, whereas the actual degree of the polynomial computed by $C$ requires more computation in order to check for cancellation of potential high-degree terms.

\subsection{Reductions}
A \emph{p-projection} \cite{valiant} over a ring $R$ from a vector $\overline{x}$ of variables to a vector $\overline{y}$ of variables is an assignment to each $y$-variable of either an $x$-variable or a constant from $R$. A p-projection is \emph{constant-free} if the only constants from $R$ used are $\{0,1,-1\}$. For a complexity class $\mathcal{C}$, a p-projection is $\mathcal{C}$-uniform if the aforementioned assignment can be computed, given $\overline{x}$, in $\mathcal{C}$. We will see logspace-uniform constant-free p-projections in \Cref{lem:algebraic}, where $\overline{x}$ is a vector representing the coefficients of an algebraic circuit and $\overline{y}$ is a vector representing coefficients of a system of equations.

In the construction of the field $\K$ (\Cref{lem:field}), we will need to say that something is constructible ``in $\cc{NP}$'', so that it can be used as a subroutine in the $\cc{NP}$ algorithm for \PIT that is being built in the proof. We formalize this as follows. A nondeterministic function is a ``function with multiple outputs'', or equivalently, a relation $f \subseteq D \times R$ (with ``domain'' $D$ and ``range'' or codomain $R$). Even though they are relations, we prefer to think of them as functions; thus, rather than writing $(d,r) \in f$, we say that $r$ is an output of $f(d)$, or write $f(d) \mapsto r$ (even though there may be more than one $r$ for a given $d$). A nondeterministic function $f$ is total if for all $d \in D$, there is at least one $r \in R$ such that $f(d)$ outputs $r$. For a finite alphabet $\Sigma$, $\cc{TFNP}$ \cite{MP} or equivalently $\cc{NPMV}_{gt}$ \cite{selman}, is the class of nondeterministic total functions where membership in the corresponding relation $f$ is decidable in $\cc{P}$, and for each $d \in D$, there exists an $r \in R$ such that $(d,r) \in f$ and  $|r| \leq \poly(|d|)$. 

\subsection{Translating between versions of SAT in IPS} \label{sec:prelim:sat}
Here we show that some standard reductions between different versions of SAT can all be efficiently simulated in IPS. We use the following standard (un-negated) algebraic translation of Boolean functions:
\begin{eqnarray*}
\alg(x) & = & x \\
\alg(\neg \varphi) & = & 1 - \alg(\varphi) \\
\alg(\varphi \wedge \psi) & = & \alg(\varphi) \alg(\psi) \\
\alg(\varphi \vee \psi) & = & 1 - (1-\alg(\varphi))(1-\alg(\psi))
\end{eqnarray*}

This translation has the property on $\{0,1\}$ assignments $\overline{\alpha}$ that $\alg(\varphi)(\overline{\alpha}) = \varphi(\overline{\alpha})$ (where on the left-hand side the output is the number $0,1 \in \F$, while on the right-hand side the value is $0$ or $1$ representing the Boolean values True and False). In particular, this means that the polynomial equation $1-\alg(\varphi)=0$ is satisfied by an input $\overline{\alpha} \in \{0,1\}^n$ iff the Boolean function $\varphi$ is satisfied by viewing $\overline{\alpha}$ as an assignment to the Boolean variables.

We begin by showing that an all-at-once algebraic translation of a CNF and a clause-by-clause algebraic translation of a CNF (resulting in a number of algebraic equations equal to the number of clauses) are equivalent from the point of view of IPS.

\begin{lemma} \label{lem:cnf_vs_clauses}
Let $R$ be a ring. Given a $k$-CNF $\varphi = \kappa_1 \wedge \dotsb \wedge \kappa_m$ in with clauses $\kappa_i$, from $\{1-\alg(\kappa_i): i = 1, \dotsc m\}$,there is a constant-free IPS$_R$ derivation of $1-\alg(\varphi)$ of size $\poly(m)$, and conversely from $1-\alg(\varphi)$ and the Boolean axioms $x_i^2 - x_i$, there is a constant-free IPS$_R$ derivation of $1-\alg(\kappa_i)$ of size $\poly(k,m)$.
\end{lemma}

This lemma also follows from the fact that the corresponding derivations can be done in low-depth Frege, and IPS simulates Frege \cite{Pitassi96, GrochowPitassi}.

\begin{proof}
From the definition of the algebraic translation, we have $\alg(\varphi) = \prod_{i=1}^m \alg(\kappa_i)$. 

We claim that 
\[
C(\overline{x}, \overline{y}) = 1 - \prod_{i=1}^m (1-y_i)
\]
is a linear-size, depth-two IPS derivation of $1-\alg(\varphi)$ from $\{1-\alg(\kappa_i) : i = 1, \dotsc, m\}$. First, if we substitute $1-\alg(\kappa_i)$ for $y_i$ for each $i=1, \dotsc, m$, then by definition of the standard algebraic translation we get $1-\alg(\varphi)$, as desired. Furthermore, if we substitute $0$ for all the $y_i$'s, then we get $C(\overline{x}, \overline{0}) = 1 - \prod_{i=1}^m (1-0) = 1 - 1 = 0$. Thus $C$ is a valid IPS certificate deriving $1-\alg(\varphi)$ from the $1-\alg(\kappa_i)$.

Conversely, suppose we start with $1-\alg(\varphi)$ and the Boolean axioms $x_i^2 - x_i$. Let our placeholder variables be $y_0$ for $1-\alg(\varphi)$, and $y_1, \dotsc, y_n$ for $x_1^2 - x_1, \dotsc, x_n^2 - x_n$. Let $C_1(\overline{x}, \overline{y})$ be a derivation of $\alg(\kappa_i)^2 - \alg(\kappa_i)$ from the Boolean axioms (this is a special case of, e.\,g.,  \cite[Lem.~3.6]{GrochowPitassi}). Then we claim that 
\[
D = (1-\alg(\kappa_i)) y_0 - C_1 \cdot \prod_{j \neq i} \alg(\kappa_j)
\]
is an IPS derivation of $1-\alg(\kappa_i)$. If we substitute in $1-\alg(\varphi)$ for $y_0$ and the Boolean axioms for $y_1, \dotsc, y_n$, then we get 
\begin{align*}
 & (1-\alg(\kappa_i))(1-\alg(\varphi)) - (\alg(\kappa_i)^2 - \alg(\kappa_i))\prod_{j \neq i}\alg(\kappa_j)  \\
 & = (1-\alg(\kappa_i))\left(1- \prod_{j=1}^m \alg(\kappa_j)\right) + (1-\alg(\kappa_i)) \alg(\kappa_i) \prod_{j \neq i}\alg(\kappa_j) \\
 & = (1-\alg(\kappa_i))\left(1- \prod_{j=1}^m \alg(\kappa_j)\right) + (1-\alg(\kappa_i)) \prod_{j=1}^m \alg(\kappa_j) \\
 & = 1 - \alg(\kappa_i).
\end{align*}

Finally, since $C_1$ was an IPS derivation by assumption, we have $C_1 \in \langle y_1, \dotsc, y_m \rangle$. Since $D$ is of the form $y_0 \cdot * + C_1 \cdot *$, it is visibly in the ideal $\langle y_0, y_1, \dotsc, y_m \rangle$, as required, and is thus a valid IPS certificate deriving $1-\alg(\kappa_i)$ from $1-\alg(\varphi)$ and the Boolean axioms.
\end{proof}

We now consider one of the standard reductions from CIRCUIT-SAT to CNF-SAT. We begin by recalling the reduction. Given a circuit $C(x_1, \dotsc, x_n)$ of size $s$, we introduce a new variable for each of the $s$ many gates, namely $x_{n+1}, \dotsc, x_{n+s}$. 
\begin{itemize}
\item If $g = \neg h$, then we add to our CNF the clauses $(x_g \vee x_h) \wedge (\neg x_g \vee \neg x_h)$.

\item If $g = h \wedge k$, then we add to our CNF the clauses $(\neg x_g \vee x_h) \wedge (\neg x_g \vee x_k) \wedge (x_g \vee \neg x_h \vee \neg x_k)$.

\item If $g = h \vee k$, then we add to our CNF the clauses $(\neg x_g \vee x_h \vee x_k) \wedge (x_g \vee \neg x_h) \wedge (x_g \vee \neg x_k)$. 

\item If $g=C$ is the output gate, then we add the clause $(x_g)$ to our CNF.
\end{itemize}
Let $r_1(\varphi)$ denote the resulting CNF Boolean formula.

\begin{lemma} \label{lem:circuit-sat}
Over any ring $R$, there is an IPS$_R$ derivation of $1-\alg(C)$ from $1-\alg(r_1(C))$ and the Boolean axioms, of size at most $O(\gates(C))$.
\end{lemma}

\begin{proof}
By Lemma~\ref{lem:cnf_vs_clauses}, we may equivalently start from $1-\alg(\kappa_i)$ where $\kappa_i$ are the clauses of the CNF $r_1(C)$. 

For each internal gate $g$, if we denote $\varphi_g$ the Boolean formula computed at $g$, then we will show by structural induction that there is a short IPS proof of $x_g - \alg(\varphi_g)$ from $\alg(r_1(C))$.
\begin{itemize}
\item If $g = \neg h$, then we have $\alg(\varphi_g) = 1 - \alg(\varphi_h)$. By assumption, there is a short IPS proof of $x_h - \alg(\varphi_h)$ from $\alg(r_1(C))$. For the two clauses added to our CNF corresponding to the gate $g$, we have
\[
1-\alg(x_g \vee x_h) = (1-x_g)(1 - x_h) \qquad 1-\alg(\neg x_g \vee \neg x_h) = x_g x_h.
\]
We underline the uses of the axioms or previously derived polynomials in the following derivation:
\begin{align*}
& -\underline{(1-x_g)(1-x_h)} + \underline{x_g x_h} - \underline{(x_h - \alg(\varphi_h))} \\
& = -1 + x_g + x_h - x_g x_h + x_g x_h - x_h + \alg(\varphi_h) \\
& =  x_g - (1 - \alg(\varphi_h)) = x_g - \alg(\varphi_g).
\end{align*}
This adds a single linear combination gate on top of previous derivations and the axioms.

\item If $g = h \wedge k$, we have $\alg(\varphi_g) = \alg(\varphi_h) \cdot \alg(\varphi_k)$, and the algebraic translations of our added clauses are:
\[
1-\alg(\neg x_g \vee x_h) = x_g(1- x_h) \qquad 1-\alg(\neg x_g \vee x_k) = x_g (1- x_k) 
\]
\[
1-\alg(x_g \vee \neg x_h \vee \neg x_k) = (1-x_g) x_h x_k.
\]
Then we have the following derivation:
\begin{align*}
 & -\underline{(1-x_g) x_h x_k} + x_k \cdot \underline{x_g (1-x_h)} + \underline{x_g (1 - x_k)} + x_k \cdot \underline{(x_h - \alg(\varphi_h))} + \alg(\varphi_h) \cdot \underline{(x_k - \alg(\varphi_k))} \\
 & = -x_h x_k + x_g x_h x_k + x_k x_g - x_g x_h x_k + x_g - x_g x_k + x_k x_h - x_k \alg(\varphi_h) + x_k \alg(\varphi_h) - \alg(\varphi_h) \alg(\varphi_k) \\
 & = x_g - \alg(\varphi_h) \alg(\varphi_k) = x_g - \alg(\varphi_g).
\end{align*}

Note that, because this is an IPS derivation, we may assume that there is a circuit ``on the side'' that computes the entire circuit $\alg(\varphi)$, and then when we need, e.\,g., $\alg(\varphi_h)$ as a coefficient in this derivation, that is simply another outgoing edge from the gate in $\alg(\varphi)$ corresponding to $h$. 

With this convention, aside from the overall additive size of $\varphi$ (which only gets added once), this derivation adds a layer of 3 product gates followed by a single linear combination gate, for 4 additional gates in total.

\item If $g = h \vee k$, we have $\alg(\varphi_g) = 1 - (1 - \alg(\varphi_h))(1 - \alg(\varphi_k))$. The algebraic translations of the corresponding CNF clauses are:
\[
1-\alg(\neg x_g \vee x_h \vee x_k) = x_g (1- x_h)(1-x_k)
\]
\[
1-\alg(x_g \vee \neg x_h) = (1-x_g)x_h \qquad 1-\alg(x_g \vee \neg x_k) = (1-x_g) x_k.
\]
Then we have the following derivation:
\begin{align*}
 & \underline{x_g(1-x_h)(1-x_k)} + (x_k-1)\underline{x_h (1-x_g)} - \underline{x_k (1-x_g)} \\
 & +  (1-x_k)\underline{(x_h - \alg(\varphi_h))} + (1-\alg(\varphi_h))\underline{(x_k - \alg(\varphi_k))}\\
 =&  x_g - x_g x_h - x_g x_k + x_g x_h x_k + x_k x_h - x_k x_h x_g - x_h + x_h x_g - x_k + x_k x_g \\
  & + x_h - \alg(\varphi_h) - x_k x_h + x_k \alg(\varphi_h) + x_k - x_k \alg(\varphi_h) -\alg(\varphi_k) + \alg(\varphi_h) \alg(\varphi_k) \\
  = & x_g - \alg(\varphi_h) - \alg(\varphi_k) + \alg(\varphi_h) \alg(\varphi_k) = x_g - \alg(\varphi_g).
\end{align*}
As in the case above, we need only incur the cost of calculating $\alg(\varphi_h)$ as a coefficient once throughout the entire proof. The remainder of this derivation adds a layer of 3 ``$1-x$'' gates, a layer of 3 product gates, and then a final linear combination gate, for a total of 7 additional gates.
\end{itemize}
Finally, at the output gate $g$, we will have derived $x_g - \alg(C)$. Since the final clause is translated to $1-x_g$, we may add these two together to get $1-\alg(C)$, as claimed.
\end{proof}

\section{From algebraic circuits to systems of equations}
In this section we show how to go from circuits that compute the identically zero polynomial to unsatisfiable systems of polynomial equations. In addition to playing a key role in the main result, this already will let us prove a version of our main result for \lang{Variety Emptiness} (rather than \UNSAT), which we do at the end of this section. 

Given a set $\mathcal{F}$ of $n$-variable polynomial equations over a ring $R$, we define their common zero locus over a ring $S \supseteq R$ by
\[
Z(\mathcal{F})_S := \{\overline{v} \in S^n : (\forall F \in \mathcal{F})[F(\overline{v})=0]\}.
\]
In this section only, we refer to $Z(\mathcal{F})$ as the (affine) scheme defined by $\mathcal{F}$. Two schemes defined by $\mathcal{F}$ and $\mathcal{F}'$, respectively are isomorphic if there is an isomorphism of rings 
\[
R[x_1,\dotsc,x_n] / \langle F : F \in \mathcal{F} \rangle \stackrel{\cong}{\to} R[x_1,\dotsc,x_n] / \langle F : F \in \mathcal{F}' \rangle.
\]

\begin{lemma} \label{lem:algebraic}
For any ring $R$, there is a logspace-uniform constant-free p-projection transforming any input $R$-algebraic circuit $C$ into a system of equations $\mathcal{F}_C$ over $R$, of degree at most 2, such that for all extension rings $S \supseteq R$, there is a bijection 
\[
\{\overline{v} \in S^n : C(\overline{v}) \text{ is invertible in } S\} \leftrightarrow Z(\mathcal{F}_C)_{S}.
\]
(In fact, this bijection is an isomorphism between the varieties or schemes $Z(1-zC) \subseteq \mathbb{A}^{n+1}$ and $Z(\mathcal{F}_C)$; indeed, $Z(\mathcal{F}_C)$ will be a graph over $Z(1-zC)$.)

Furthermore, when $C \equiv 0$, there is an IPS$_{R}$ certificate for the unsolvability of $\mathcal{F}_C$ computable by a circuit whose number of wires is $O(\wires(C))$, whose depth is at most $\depth(C)+2$, and where the only constants used in the IPS proof are $0,1,-1$ and those used in $C$.
\end{lemma}

While increasing the depth by 2 may seem like a significant cost in terms of depth (especially for low-depth circuits), we note that in fact the IPS certificate has the form $z \cdot F + y_i$ where $\depth(F) \leq \depth(C)$ (see the last paragraph of the proof), so the additional 2 in the depth is only for multiplying by a variable, then adding another variable.

\begin{proof}[Proof idea for Lemma~\ref{lem:algebraic}]
We build up equations that simulate the circuit $C$ gate by gate, and then add one additional variable and equation saying that $C$ is nonzero, namely $1-zC=0$. 
\end{proof}

\begin{proof}[Proof of Lemma~\ref{lem:algebraic}]
Let $C$ be an $R$-algebraic circuit. For each gate $v$ of $C$ we will have a new variable $g_v$, as well as one additional variable $z$. The system of equations  $\mathcal{F}_C$ will be
\begin{eqnarray*}
g_v - x & & \text{ if $v$ is an input gate with variable $x$} \\
g_v - c & & \text{ if $v$ is an input gate with constant $c$} \\
g_v - \sum_{i=1}^k \alpha_i g_{u_i} & & \text{ if $v$ is a linear combination gate $v = \sum \alpha_i u_i$} \\
g_v - g_u g_w & & \text{ if $v$ is a product gate $v = u w$} \\
g_vg_u - 1 & & \text{ if $v$ is an inversion gate $v = 1/u$} \\
1 - z g_{C} & & \text{ for the output gate $g_C$}
\end{eqnarray*}
(Recall that inverseion gates are only allowed when the input is syntactically constant and invertible.) Note that the only constants used in $\mathcal{F}_C$ are $\pm 1$ and the constants already present in $C$; it is clear that this is a logspace-uniform p-projection.

The bijection from the set of inputs that make $C$ evaluate to an invertible value to the variety $Z(\mathcal{F}_C)$ is given as follows. For each gate $v$, let $f_v$ denote the polynomial computed at $v$, and we write $f_C$ for the polynomial computed by $C$. Given an input vector $\overline{\alpha} \in S^n$ such that $C(\overline{\alpha})$ is invertible, we assign each gate variable $g_v$ the value of $f_v(\overline{\alpha})$. Since we assumed $C(\overline{\alpha})$ is invertible, we assign $z$ its inverse. The map in the other direction is to take any solution $(\overline{x}, \overline{g}, z)$ to $\mathcal{F}_C$ and simply consider the $\overline{x}$ part of the vector. This map is injective because the values of the $g$ variables and $z$ are uniquely determined by $\overline{x}$.

(To see that this map is an isomorphism of schemes, we show what it corresponds to at the level of rings. The ring homomorphism $R[\overline{x}, z] / \langle 1 - zf_C \rangle \to R[\overline{x}, \overline{g}, z] / \langle \mathcal{F}_C \rangle$ is simply the inclusion of subrings. In the opposite direction, $\overline{x}$ and $z$ are mapped to themselves, while each variable $g_v$ gets mapped to the polynomial $f_v(\overline{x})$. It is readily verified that these are mutually inverse ring homomorphisms.)

Finally, we will exhibit the claimed IPS$_R$ proof that $\mathcal{F}_C$ is unsatisfiable when $C \equiv 0$. The key is the following claim:

\begin{claim} From the above equations $\mathcal{F}_C$, IPS can derive $g_v - f_v$ for every gate $v$ in the circuit by a derivation whose number of wires depends linearly on $\wires(C)$ and whose depth is at most that of $C$. \end{claim}

We proceed by structural induction. 

\begin{itemize}
\item If $v$ is an input gate with variable $x$ (resp., constant $c$), then $f_v = x$ (resp., $f_v = c$), and the equation $g_v - x$ (resp., $g_v - c$) is one of the equations in $\mathcal{F}_C$. 

\item If $v = \sum \alpha_i u_i$ is a linear combination gate, then suppose by induction that we have derived $g_{u_i} - f_{u_i}$ for all $i$. Then starting from the axiom $g_v - \sum_{i=1}^k \alpha_i g_{u_i}$, we add $\sum \alpha_i (g_{u_i} - f_{u_i})$, and the result is then $g_v - \sum \alpha_i f_{u_i} = g_v - f_v$.

\item If $v = u w$ is a product gate, suppose by induction we have derived $g_{u} - f_{u}, g_w - f_w$. Then we derive $g_v - f_v = g_v - f_{u} f_{w}$ as follows, where the first line exhibits this in terms of the previously derived polynomials and the equations of $\mathcal{F}_C$:
\begin{align*}
 &  (g_v - g_u g_w) + (g_u - f_u)(g_w - f_w) + f_w( g_u - f_u) + f_u (g_w - f_w)  \\
 = & (g_v - g_u g_w) + (g_u g_w - f_u g_w - f_w g_u + f_u f_w) + (f_w g_u - f_w f_u) + (f_u g_w - f_u f_w) \\
 = & g_v - f_w f_u = g_v - f_v.
\end{align*}
To get the coefficients $f_w, f_u$ in the preceding derivation, the IPS certificate contains a copy of the circuit $C$, and simply uses the output of the gate $u$ (resp., $w$) to get $f_u$ (resp., $f_w$). The rest of the IPS certificate follows the above induction.

\item If $v = 1/u$ is an inversion gate, suppose by induction we have derived $g_u - f_u$. Recall that, by assumption, $f_u$ is an invertible constant. We use the following linear combination
\[
\frac{1}{f_u}\left[ (g_v g_u - 1) - g_v(g_u - f_u) \right]
\]
to get $g_v - 1 / f_u = g_v - f_v$, as desired.
\end{itemize}

Thus, there is an IPS derivation with number of wires $O(\wires(C))$ and the same depth as that of $C$, that derives $g_C - f_C$. This completes the proof of the claim.

Now, if $C \equiv 0$, then $f_C = 0$, so $g_C-f_C$, which IPS efficiently derived, is in fact just $g_C$. Then using the final equation we derive $z \cdot g_C +  (1 - z g_C) = 1$. This step increases the depth by 2, and shows that $1$ is in the ideal $\langle \mathcal{F}_C \rangle$.
\end{proof}

\begin{remark}
We note that the above proof uses the full power of circuit-based IPS, in the sense that the pattern of re-use of the derived equation $g_v - f_v$ is nearly the same in the IPS certificate as the pattern of re-use of the output of the gave $v$ in the original circuit $C$. 
\end{remark}

We now come to the version of our main theorem for \lang{Variety Emptiness} (rather than \UNSAT).

\begin{theorem} \label{thm:mainpoly}
Let $\F$ be a finite field or an algebraic number field. Measuring IPS proof size by total bit-size, we have:

\begin{centering}
There exists a Cook--Reckhow proof system that p-simulates IPS$_\F$ for \lang{Variety Emptiness} over $\F$ \\

$\Longleftrightarrow$ \\

$\PIT_{\F}$ is in $\cc{NP}$.

\end{centering}
\end{theorem}

\begin{proof}
($\Leftarrow$) If PIT over $\F$ is in $\cc{NP}$, then IPS (with size measured as total bit-size) has $\cc{NP}$-verifiable proofs. Our Cook--Reckhow proof system takes as input an IPS certificate $F(\overline{x}, \overline{y})$, together with the two $\cc{NP}$ certificates for the two instances of PIT $F(\overline{x}, \overline{0}) = 0$ and $F(\overline{x}, \mathcal{F}_C)-1=0$. This is directly seen to p-simulate IPS.

($\Rightarrow$) Suppose there is a Cook--Reckhow proof system $\mathcal{P}$ that p-simulates IPS for \lang{Variety Emptiness} over $\F$. We give an $\mathsf{NP}$ algorithm for PIT over $\F$. Given an algebraic circuit $C$ of bit-size $n$ over $\F$ consider the following steps. From Lemma~\ref{lem:algebraic}, in logspace we then construct a system of equations $\mathcal{F}_C$ over $\F$ such that $C$ is the identically zero polynomial iff $\mathcal{F}_C$ is unsatisfiable over the algebraic closure $\overline{\F}$. Furthermore, when $\mathcal{F}_C$ is unsatisfiable, there is an IPS proof of this whose size is $\poly(n)$. Then, by assumption, $\mathcal{F}_C$ thus also has a short refutation in the Cook--Reckhow system $\mathcal{P}$. The $\mathsf{NP}$ machine now guesses and verifies a short $\mathcal{P}$-refutation for $\mathcal{F}_C$. If it finds one, it returns YES (the circuit $C$ was identically zero). Otherwise it returns NO.
\end{proof}

\section{From nonzero polynomials to nonzero functions over a finite field}
In this section we lay out the definitions and machinery that allow us to go from nonzero polynomials over a field to nonzero functions over some finite field, encapsulated  in Lemma~\ref{lem:field}.

\begin{lemma} (From nonzero polynomial to nonzero function over a finite field) \label{lem:field}
Let $\F$ be a finite field or the rational numbers. 
Let $C$ be a $\F$-algebraic circuit of total bit-size $s$, computing a polynomial of degree $d$. Then there is a finite field $\K$ of order at most $O(\max\{2^s, d\})$ such that (1) evaluating $C$ over inputs from $\K$ is well-defined,\footnote{While this notion hopefully makes intuitive sense, we can formalize it as follows. Let $p$ be the characteristic of $\F$; so $p$ is either $0$ or a prime, and in the former case we have $\Z/p\Z=\Z/0\Z=\Z$. Let $K$ be the set of coefficients appearing in all polynomials at all gates of $C$, and let $(\Z/p\Z)[K]$ be the ring over $\Z/p\Z$ generated by $K$. Then when we say ``evaluating $C$ over inputs from $\K$ is well-defined,'' what we mean is that there is a ring homomorphism $\varphi\colon (\Z/p\Z)[K] \to \K$ that sends $1$ to $1$, and we consider $C$ as a circuit over $\K$ by applying $\varphi$ to all constants and gates in the circuit.} and (2) the function $\K^n \to \K$ computed by $C$ is not the zero function if and only if $C$ is not identically zero as a formal polynomial, and (3) $\K$ can be constructed in $\cc{TFNP}=\cc{NPMV}_{gt}$.

In the case of characteristic zero, we may take $\K = \F_p$ for any prime $p > \max\{2^s, d\}$, and if $\F = \F_q$ with $q$ a prime power, we may take $\K = \F_{q^e}$ for any $e$ such that  $q^e > d$.
\end{lemma}

In particular, if $C$ is a constant-free circuit over $\Q$, then any prime $p > \sdeg(C)$ works, by Observation~\ref{obs:sdeg}. In the case of finite fields, we may take any $e > \log_2 \sdeg(C)$.

\begin{proof}[Proof of Lemma~\ref{lem:field}]
We handle characteristic zero and positive characteristic separately. 

\textit{Characteristic zero.} In characteristic zero, for each gate $v$, let $f_v$ denote the polynomial over $\Q$ computed at gate $v$. Let $D(f_v)$ be the set of integers appearing as the denominators of the coefficients of $f_v$, when each coefficient is put into reduced form (the numerator and denominator are coprime). Let $D = \bigcup_{v \in C} D(f_v)$. Then $C$ computes a well-defined function over $\mathbb{Z}/m\mathbb{Z}$ for any integer $m$ such that every element of $D$ is coprime to $m$.

Now we show that there exists a prime $p$ that is coprime to every element of $D$, has bit-length at most $\poly(|C|)$, and such that the function $\F_p^n \to \F_p$ computed by $C$ is not the zero function if and only if $C$ is not identically zero as a polynomial. Since the bit-length of any constant appearing in $C$ is at most $s$ by assumption, if $p > 2^s$ then all of the nonzero constants appearing in $C$ are coprime to $p$, and have well-defined and nonzero reductions modulo $p$. 

If $C$ is the identically zero polynomial, then clearly $C$ computes the identically zero function on $\F_p^n$. Conversely, if $C$ is not identically zero as a polynomial, then for $p > d$ (the degree of the polynomial computed by $C$), the PIT Lemma~\ref{lem:pit} implies that the function  $\F_p^n \to \F_p$ computed by $C$ is nonzero.

By Bertrand's Postulate, there exists a prime $p$ in between $\max\{2^s,d\}$ and $2\max\{2^s, d\}$, and thus the order of $\F_p$ is at most $2\max\{2^s, d\}$ and $C$ computes a nonzero function over $\F_p$. This proves existence for the case of characteristic zero.

To find such a $p$ constructively, a nondeterministic machine can guess a bit-string of length $\lceil \log_2 \max\{2^s, \sdeg(C)\} \rceil$, prepend it with a 1, and then verify whether the integer encoded by that bit-string is prime. 

\textit{Positive characteristic.} In the case of a finite field of order $q$, we move to an extension field of $\F_q$ to find a nonzero function. As in the case of characteristic zero, once $q^e > d$, the PIT Lemma~\ref{lem:pit} implies that the function computed by $C$ over $\F_{q^e}$ is nonzero iff $C$ is not identically zero as a formal polynomial. It thus suffices to take any integer $e$ such that $q^e > d$. Since $d \leq \sdeg(C) \leq 2^{\proddepth(C)}$, if we take  $e = \lceil \log_q \sdeg(C) \rceil+1$, we then have that $\log_2 |\F_{q^e}| \leq O(\proddepth(C)) \leq O(|C|)$). This completes the proof of existence.

In order to efficiently construct $\K = \F_{q^e}$, we will need to be able to construct an irreducible polynomial of degree $e$ over $\F_q$ . We only need to do this in an $\cc{NP}$ fashion, which can be done by nondeterministically guessing a polynomial of degree $e$ and then verifying that it is irreducible in time $\poly(e, \log q)$ with the standard algorithm (see, e.\,g., \cite[Thm.~14.37]{GGbook}. 
\end{proof}

\begin{remark}[On the level of constructivity]
For our results, we only needed an upper bound of $\mathsf{NPMV}_{gt}$ on constructing $\K$, but in fact we can do quite a bit better.

In the case of characteristic zero, finding a prime $p$ in the right range can be done in Las Vegas randomized time (expected polynomial time with zero error, which one might call $\mathsf{ZPPMV}_{gt}$), rather than just $\mathsf{NPMV}_{gt}$, as follows. A random $n$-bit number is prime with probability $\sim 1/n$ (essentially equivalent to the Prime Number Theorem), and after guessing a random number of $n$ bits, primality can be verified in $\mathsf{P}$ \cite{AKS}. This can be repeated until a prime is found, which on expectation happens after $O(n)$ trials. Derandomizing algorithms for constructing primes is a well-known open question.

In the case of characteristic $p$, finding an irreducible polynomial over $\F_p$ whose degree is in the right range can in fact be done deterministically in polynomial time, using either of the following theorems:

\begin{theorem}[{Shoup \cite{shoup}}] \label{thm:shoup}
Given as input a prime $p$ and target degree $e$, there is an algorithm that constructs a polynomial of degree $e$ that is irreducible over $\F_p$ in time $\tilde O(p^{1/2} e^4)$, where the $\tilde O$ hides factors polynomial in $\log p$ and $\log e$. 
\end{theorem}

\begin{theorem}[{Adleman \& Lenstra \cite{AL86}}] \label{thm:AL86}
There is a constant $c > 0$ and an algorithm which, given as input a prime $p$ and target degree $e$, constructs a polynomial of degree in the range $\left[\frac{e}{(c \log p)}, e\right]$ that is irreducible over $\F_p$ in time $O((e \log p)^c)$.
\end{theorem}

\end{remark}

\section{Simulating circuits over finite fields by Boolean circuits with short IPS proofs}
Let $\F_q$ be a finite field of order $q=p^e$ with $p$ prime and $q \leq 2^{\poly(n)}$. We will show how to simulate arithmetic over $\F_q$ via CNFs, in such a way that IPS has short derivations of the standard translations of the CNFs. This is similar to the encodings of arithmetic used in Santhanam \& Tzameret \cite{SanthanamTzameret} and in Alekseev, Grigoriev, Hirsch, \& Tzameret \cite{AGHT}, with some key differences. In the former, they simulate arithmetic over finite fields but using a \emph{unary} encoding, whereas in our case we need to use the binary encoding because our field size can be exponential. In the latter, they simulate arithmetic over $\Q$ using a binary encoding, but do not need to deal with taking the numbers modulo some prime. In this section we extend the machinery of \cite{AGHT} to handle arithmetic in positive characteristic, which essentially improves on \cite{SanthanamTzameret} by using a binary encoding. Our key addition here is a Boolean circuit implementing the remainder-mod-$p$ operator, yet that IPS can still reason about efficiently.

\subsection{From circuits over finite fields to circuits over finite \emph{prime} fields}
In this section, we recall the following standard construction and lemma to reduce the case of $\F_{p^e}$ to the case of $\F_p$ with $p$ prime. Note that, since we assume $p^e \leq 2^{\poly(n)}$, that we have $e \leq \poly(n)$. 

\begin{definition}[Vectorization of circuits over a field extension]
Suppose $\K$ is a degree-$e$ extension field of a field $\F$, and let $\iota \colon \K \to \F^e$ be an $\F$-linear bijection such that $\iota^{-1}(e_1)=1$, where $e_1=(1,0,\dotsc,0)$. We define $\overline{\VEC}_{\K/\F}$ and $\VAL_{\K/\F}$ as follows:
\begin{itemize}
\item Given $\overline{v} \in \F^e$, we define $\VAL_{\K/\F}(\overline{v}) := \iota^{-1}(\overline{v})$.

\item For a vector of variables $x_{i,1}, \dotsc, x_{i,e}$ over $\F$, we define $\VAL_{\K/\F}(\overline{x}_i) := \sum_{j=1}^e \iota^{-1}(e_j) x_{i,j}$, where $e_j \in \F^e$ is the $j$-th standard basis vector

\item For $\alpha \in \K$, $\overline{\VEC}_{\K/\F}(\alpha) := \iota(\alpha) \in \F^e$.

\item For a variable $x_i$ over $\K$, we define $\overline{\VEC}_{\K/\F}(x_i) := (x_{i,1}, \dotsc, x_{i,e})$ where each $x_{i,j}$ is a new variable over $\F$

\item For a linear combination gate $F = \sum_{i=1}^k \alpha_i G_i$ over $\K$, we proceed as follows. For any $\alpha \in \K$, we define $L_\alpha \colon \F^e \to \F^e$ by the composition:
\[
\underbrace{\F^e \stackrel{\iota^{-1}}{\to} \K \stackrel{\times \alpha_i}{\to} \K \stackrel{\iota}{\to} \F^e}_{L_{\alpha}}.
\]
Then we define $\overline{\VEC}_{\K/\F}(F) := \sum_{i=1}^k L_{\alpha_i} (\overline{\VEC}_{\K/\F}(G_i))$, where the sum here is sum as elements of $\F^e$, i.e. $e$ parallel independent addition gates over $\F$. 
Since each $L_{\alpha_i}$ is $\F$-linear, these can be combined with the sum $\sum_{i=1}^k$ into a single linear combination gate.

\item We define $\PROD_{\K/\F}(\overline{x}, \overline{y})$ as follows. The linear maps $L_\alpha$ above give an $e$-dimensional representation of rings $\lambda \colon \K \to M_e(\F)$, by $\lambda(\alpha) := L_\alpha$. Since we have assumed $\iota^{-1}(e_1)=1$, it follows that the first column of the $e \times e$ matrix $L_\alpha$ is the same as $\iota(\alpha)$, i.\,e., the vector representation of the corresponding element of $\K$. Let $\pi\colon M_e(\F) \to \F^e$ be the projection onto the first column. Then we define $\PROD_{\K/\F}(\overline{x}, \overline{y})$ to be a depth-3 circuit over $\F$ that takes in $2e$ inputs and has $e$ outputs, and implements the following composition:
\[
\F^e \times \F^e \stackrel{\iota^{-1} \times \iota^{-1}}{\to} \K \times \K \stackrel{\lambda \times \lambda}{\to} M_e(\F) \times M_e(\F) \stackrel{\text{mult}}{\to} M_e(\F) \stackrel{\pi}{\to} \F^e.
\]

The key non-trivial part here is the $e \times e$ matrix multiplication, which can be implemented by the usual depth-2 circuit of size $O(e^3)$, viz. $(A \cdot B)_{ij} = \sum_{k=1}^e A_{ij} \cdot B_{jk}$. (The depth of 3 comes from precomposing this circuit with the linear maps $\lambda \circ \iota^{-1}$. Post-composing with the linear projection $\pi$ can be absorbed into the linear combination gate $\sum_{k=1}^e$ without changing the depth.)

For a multiplication gates $F = G \times H$ over $\K$, we define 
\[
\overline{\VEC}_{\K/\F}(F) := \PROD_{\K/\F}(\overline{\VEC}_{\K/\F}(G), \overline{\VEC}_{\K/\F}(H)).
\]

\item We similarly define $\INV_{\K/\F}(\overline{x})$ as an $e$-input, $e$-output, depth-3 circuit over $\F$ that implements the following composition:
\[
\F^e \stackrel{\iota^{-1}}{\to} \K \stackrel{\lambda}{\to} M_e(\F) \stackrel{inv}{\to} M_e(\F) \stackrel{\pi}{\to} \F^e.
\]
If $F = 1/G$ is an inversion gate (recall: these are only allowed when $G$ is syntactically constant and invertible), then we define $\overline{\VEC}_{\K/\F}(F) := \INV_{\K/\F}(\overline{\VEC}_{\K/\F}(G))$.
\end{itemize}
\end{definition}

\begin{lemma}[Vectorization lemma for circuits over field extensions] \label{lem:extension}
Suppose $C$ is an algebraic circuit over $\K$, and $\K$ is a degree-$e$ extension field of a field $\F$, and let $\iota \colon \K \to \F^e$ be an $\F$-linear bijection such that $\iota^{-1}(e_1)=1$. 
Let $C' := \overline{\VEC}_{\K/\F}(C)$.
\begin{enumerate}
\item $C'$ can be constructed from the description of $C$ in logspace.

\item The following is a polynomial identity, hence has a trivial IPS$_{\K}$ derivation from no axioms:
\[
\VAL_{\K/\F}(\overline{\VEC}_{\K/\F}(C)) \equiv C(\VAL_{\K/\F}(\overline{\VEC}_{\K/\F}(\overline{x}))).
\]

\item The various size measures of $C'$ can be bounded as follows:
\begin{align*}
\wires(C') & \leq e \cdot \wires(C) + O(e^3) \cdot \gates(C) \\
\gates(C') & \leq O(e^3) \cdot \gates(C) \\
\depth(C') & \leq 3 \cdot \depth(C) \\
\proddepth(C') & =  \proddepth(C)
\end{align*}

\end{enumerate}
\end{lemma}


\begin{remark}
If, in the above observation, we desired a circuit $C''$ over $\F$ with a single output that was guaranteed to be a nonzero function when $C \not\equiv 0$, instead of $C'$ with multiple outputs, it is actually a little bit tricky. (This will not be an issue for our application, but is a natural question.) Certainly one of the output gates of $C'$ computes a nonzero function over $\F$. But how to tell which one? More flexibly, there exists an $\F$-linear combination of the $e$ output gates that will produce a nonzero function (in fact, a positive fraction of all such linear combinations will have this property), but how to find one deterministically? Although we do not know how to answer this question, we note that this problem need \emph{not} be as hard as finding a nonzero linear combination among $e$ polynomials, when such a linear combination is guaranteed to exist. For, in our case, the $e$ polynomials we are looking at are guaranteed to be the $e$ coordinates of a single algebraic circuit over $\K$, whose size is comparable to the size of the $e$-output $\F$-circuit we ended up with.
\end{remark}

\subsection{From circuits over finite prime fields to Boolean circuits}
We now focus on the case $\F=\F_p$ with $p > 0$ prime; by our assumption, we have $p \leq 2^{\poly(n)}$. For consistency, we follow the notation of \cite[Section~5]{AGHTarxiv}. Whereas they worked in the two's complement notation over $\Z$, we will (mostly) work with ordinary binary encoding for $\{0, \dotsc, p-1\}$ over $\Z/p\Z$, encoded into $b_p := \lfloor \log_2 p \rfloor + 1$ bits; note that $b_p \leq \poly(n)$. We will essentially work over $\Z$, but we will take our results modulo $p$ after each operation (addition or multiplication). When we take a product of two elements of $\Z$ in $\{0, \dotsc, p-1\}$, it may be as large as $p^2$, hence require $\sim 2b_p$ bits to represent, but this is still $\poly(n)$.

We use the $\VAL$ operation (which takes a bit-vector to the corresponding integer it represents in two's complement notation), and the $\ADD, \PROD, \BIT$ Boolean circuits from \cite{AGHTarxiv}. We define the $\VAL_+$ operation to be an ``unsigned'' version of their $\VAL$ operation, in ordinary (rather than two's-complement) binary notation, that is, 
\[
\VAL_+(\overline{x}) := \sum_{i=0}^{b-1} 2^i x_i,
\]
where $b$ is the number of bits in (i.\,e., the length of) the bit-vector $\overline{x}$.

\begin{definition}[The binary value mod $p$ operation $\VAL_p$]
Given a bit vector $x_{b-1} x_{b-2} \dotsc x_0$, denoted $\overline{x}$, we define the following algebraic circuit with $\Z/p\Z$ coefficients:
\[
\VAL_p(\overline{x}) := \sum_{i=0}^{b-1} 2^i x_i \pmod{p}.
\]
Note that here we are working in $\Z/p\Z$, so each $2^i \pmod{p}$ is in fact an element of $\{0,\dotsc,p-1\}$, as is the entire sum.
\end{definition}

\begin{lemma} \label{lem:rem}
For each $b \geq 1$, there is a Boolean circuit $\overline{\REM}_p(\overline{x})$ with $b$ variables and $b_p$ outputs such that 
\begin{enumerate} 
\item $0 \leq \VAL_+(\overline{\REM}_p(\overline{x})) < p$ for all $\overline{x}$ and $\VAL_+(\overline{\REM}_p(\overline{x})) \equiv \VAL_+(\overline{x}) \pmod{p}$.

\item IPS$_{\F_p}$ efficiently derives $\VAL_p(\overline{\REM}_p(\overline{x})) - \VAL_p(\overline{x})$ from the Boolean axioms.

\item Given $b,p$ as input, $\overline{\REM}_p$ can be constructed in time $\poly(b, \log_2 p)$ (in particular, it has size at most $\poly(b,\log_2 p)$)
\end{enumerate}
\end{lemma}

We suppress the dependence on $b$ from the notation, since it is implied by the number of bits in the bit-vector $\overline{x}$. 

\begin{proof}[Proof idea]
The circuit essentially implements the following pseudo-code:
\begin{verbatim}
for i = len(x)-1 down to 0 do
  if 2^i * p <= x then:
    x = x - 2^i * p
  end if
end for
return x
\end{verbatim}
The loop is unrolled into a circuit. Each iteration of the loop corresponds to a part of the circuit that looks intuitively like
\begin{center}
 ``if $x < 2^i p $ then $x$, else $x - 2^i p$''. 
 \end{center}
 Although it is open whether or not IPS can reason about \emph{in}equalities---in general this can be done iff IPS is equivalent to the Cone Proof System \cite{AGHT}---in this case we are able to get around this potential issue as follows. 

It is the use of the conditional actually helps save us from having to really deal with inequalities, as follows. As a Boolean circuit, the conditional is syntactically of the form $(\varphi \wedge A) \vee (\neg \varphi \wedge B)$, where $\varphi$ encodes the inequality $x < 2^i p$, $A$ encodes continuing to use $x$ and running the remaining $i-1$ iterations, and $B$ encodes replacing $x$ by $x - 2^i p$ and then running the remaining $i-1$ iterations. The algebraic translation of this has the form:
\[
\alg(\varphi) \alg(A) + (1-\alg(\varphi)) \alg(B) + \alg(\varphi)(1-\alg(\varphi))\alg(A)\alg(B).
\]
Since $A$ and $B$ have polynomial-sized Boolean circuits, their translations have polynomial-size algebraic circuits. The final term here is cancelable by deriving $\alg(\varphi)^2 - \alg(\varphi)$ from the Boolean axioms---since $\varphi$ has a small Boolean circuit---and then multiplying that by the polynomial-sized circuit $\alg(A)\alg(B)$. 

For the remaining two terms, as long as IPS can prove that both $\alg(A)$ and $\alg(B)$ are equal to $\VAL_p(\overline{x})$, it can efficiently derive that the whole thing has value equal to $\VAL_p(\overline{x})$, viz.:
\begin{align*}
& \alg(\varphi) (\alg(A) - \VAL_p(\overline{x})) + (1-\alg(\varphi)) (\alg(B) - \VAL_p(\overline{x})) \\
= & (\alg(\varphi) \alg(A) + (1-\alg(\varphi)) \alg(B)) - \VAL_p(\overline{x}).
\end{align*}

Finally, we mention an issue that is not technically needed for correctness, and is obvious in retrospect, but was a stumbling block in discovering the proof. We note that the for loop does not need to know ``when to stop subtracting multiples of $p$''. For, if at any point, the value of $x$ has come to lie in the interval $\{0,\dotsc,p-1\}$, whenever $x_i = 1$ after that, we have that $x < p$, so the remaining iterations will simply have no effect on $x$, as desired.
\end{proof}

Before coming to the proof, the proof idea above has already highlighted the need for the following lemma, as we will apply it repeatedly for each iteration of the \texttt{for} loop.

\begin{lemma}[Composition Lemma] \label{lem:composition}
Let $\overline{C}(\overline{x}), \overline{D}(\overline{x})$ be two $b$-input, $b$-output Boolean circuits, where $D$ has size $s$. Suppose there are IPS derivations of 
\[
\VAL_p(\overline{C}(\overline{x})) - \VAL_p(\overline{x}) \text { and } \VAL_p(\overline{D}(\overline{x})) - \VAL_p(\overline{x})
\]
from the Boolean axioms, of sizes $t$ and $r$, respectively. Then there is an IPS derivation of
\[
\VAL_p(\overline{C}(\overline{D}(\overline{x}))) - \VAL_p(\overline{x})
\]
from the Boolean axioms, of size $t+r+\poly(s)$.
\end{lemma}

\begin{proof}
In addition to the IPS certificates stated, we will also need the IPS derivation of the Boolean equations for $\overline{D}(\overline{x})$ (that is, $D_i(\overline{x})^2 - D_i(\overline{x})$) from the Boolean axioms \cite[cf. Lem.~3.6]{GrochowPitassi}. We call these three certificates $F,G,H$, that is, we have
\begin{align*}
F(\overline{x}, \overline{0}) & =  0 & 
G(\overline{x}, \overline{0})  & = 0 \\
F(\overline{x}, \overline{x}^2 - \overline{x}) &  = \VAL_p(\overline{C}(\overline{x})) - \VAL_p(\overline{x}) & 
G(\overline{x}, \overline{x}^2 - \overline{x}) &  = \VAL_p(\overline{D}(\overline{x})) - \VAL_p(\overline{x}) \\
\end{align*}
\begin{align*}
H_i(\overline{x}, \overline{0}) & = 0  \\
H_i(\overline{x}, \overline{x}^2 - \overline{x}) & = D_i(\overline{x})^2 - D_i(\overline{x}).
\end{align*}
where $F$ has size $t$, $G$ has size $r$, and $H$ has size $\poly(s)$. 

We claim that 
\[
J(\overline{x}, \overline{y}) := F(\overline{D}(\overline{x}), \overline{H}(\overline{x}, \overline{y})) + G(\overline{x}, \overline{y})
\]
is an IPS derivation of $\VAL_p(\overline{C}(\overline{D}(\overline{x})) - \VAL_p(\overline{x})$ of size $t + \poly(s) + r + 1$. The circuit $\overline{D}$ has size $s$, the circuit $\overline{H}$ has size $\poly(s)$ \cite[Lem.~3.6]{GrochowPitassi}, $F$ has size $t$, $G$ has size $r$, and we incur one additional gate for the addition between $F$ and $G$, but that $+1$ can be absorbed into the $\poly(s)$ summand.

To see that $J$ is an IPS derivation as claimed, we calculate:
\begin{align*}
J(\overline{x}, \overline{0})  & = F(\overline{D}(\overline{x}), \overline{H}(\overline{x}, \overline{0})) + G(\overline{x}, \overline{0}) \\
& =  F(\overline{D}(\overline{x}), \overline{0}) + 0 & \text{(since $H_i, G$ are IPS certificates)} \\
& =  0 & \text{(since $F$ is an IPS certificate)}
\end{align*}
Finally, we have
\begin{align*}
J(\overline{x}, \overline{x}^2 - \overline{x}) & = F(\overline{D}(\overline{x}), \overline{H}(\overline{x}, \overline{x}^2-\overline{x})) + G(\overline{x}, \overline{x}^2-\overline{x}) \\
 & = F(\overline{D}(\overline{x}), \overline{D}(\overline{x})^2 - \overline{D}(\overline{x})) + \left(\VAL_p(\overline{D}(\overline{x})) - \VAL_p(\overline{x})\right) \\
 & =  \left(\VAL_p(\overline{C}(\overline{D}(\overline{x}))) - \VAL_p(\overline{D}(\overline{x})) \right) + \left(\VAL_p(\overline{D}(\overline{x})) - \VAL_p(\overline{x})\right) \\
 & =  \VAL_p(\overline{C}(\overline{D}(\overline{x}))) - \VAL_p(\overline{x}).\\
\end{align*}
Here, the second lines follows from the above properties of $H_i$ and $G$, and the third line follows from substituting in $\overline{D}(\overline{x})$ for $\overline{x}$ in $F(\overline{x}, \overline{x}^2-\overline{x})$.
\end{proof}

\begin{corollary} \label{cor:composition}
Suppose $\overline{C}_1(\overline{x}), \dotsc, \overline{C}_k(\overline{x})$ are $k$ Boolean circuits each with $b$ inputs and $b$ outputs, each of size at most $s$. Suppose there are IPS derivations of
\[
\VAL_p(\overline{C}_i(\overline{x})) - \VAL_p(\overline{x})
\]
of size $t_i$ from the Boolean axioms. Then there is an IPS derivation of
\[
\VAL_p(\overline{C}_1(\overline{C}_2(\dotsb (\overline{C}_k(\overline{x})))) - \VAL_p(\overline{x})
\]
of size $\sum_{i=1}^k t_i + k \cdot \poly(s)$.
\end{corollary}

\begin{proof}
By induction on $k \geq 2$. The base case is Lemma~\ref{lem:composition}. Now suppose $k > 2$ and the result holds for $k$; we will show it holds for $k+1$. By induction, we have that there is an IPS derivation of $\VAL_p(\overline{C}_1(\overline{C}_2(\dotsb (\overline{C}_k(\overline{x})))) - \VAL_p(\overline{x})$ of size $\sum_{i=1}^k t_i + \poly(k, s)$. By assumption, there is an IPS derivation of $\VAL_p(\overline{C}_{k+1}(\overline{x})) - \VAL_p(\overline{x})$ of size $t_{k+1}$, and $\overline{C}_{k+1}(\overline{x})$ has size at most $s$. By applying Lemma~\ref{lem:composition} with $C = C_1 \circ \dotsb \circ C_k$ and $D = C_{k+1}$, we get that there is an IPS derivation of $\VAL_p(\overline{C}_1(\overline{C}_2(\dotsb (\overline{C}_{k+1}(\overline{x})))) - \VAL_p(\overline{x})$ from the Boolean axioms of size $\sum_{i=1}^k t_i + k \cdot \poly(s) + t_{k+1} + \poly(s) = \sum_{i=1}^{k+1} t_i + (k+1) \cdot \poly(s)$, as claimed.
\end{proof}

\begin{observation}
Given $i$ and $p$, a Boolean circuit $\LT_{i,p}$ can be constructed in time $\poly(i,\log p)$ such that $\LT_{i,p}(\overline{x})=1$ if and only if $\VAL_+(\overline{x}) < 2^i p$.
\end{observation}

\begin{proof}
In time $\poly(i,\log_2 p)$ we can calculate the binary representation $\overline{\BIT}(2^i p)$, by taking the binary representation of $p$ and shifting it to the left (padding with zeros in lower-order bits) by $i$ bits. We then use a standard circuit for comparing integers, $\LT(\overline{x}, \overline{y})$. Because of the structure of the proof of Lemma~\ref{lem:rem} as outlined above, the details of $\LT$ won't be important, so long as it has polynomial size.
\end{proof}

Now we come to the proof of Lemma~\ref{lem:rem} for the remainder circuit $\REM_p$.

\begin{proof}[Proof of Lemma~\ref{lem:rem}]
We formalize the above proof idea.

We define the subcircuit that we will use to do the conditional subtraction as:
\[
\overline{\SUBPROD}_p(\overline{x}, i) := \overline{\ADD}(0\overline{x}, \overline{\PROD}(\overline{\BIT}(-2^i), \overline{\BIT}(p))),
\]
and we use $\SUBPROD_{p,j}$ to denote the $j$-th bit of the output of $\overline{\SUBPROD}_p$. We note that inside $\SUBPROD$, the operations $\BIT$, $\ADD$, and $\PROD$ use the two's-complement representation, as in \cite{AGHT}, in order to handle the necessary subtraction. This is why the $\overline{x}$ argument to $\SUBPROD$ is prepended with 0 (the sign bit) before being passed to $\overline{\ADD}$. In our applications of $\SUBPROD$, we will be guaranteed that both the input and output are non-negative, and will simply never use the sign bit (even though it is crucial \emph{inside} the $\SUBPROD$ circuit in order to do the subtraction).

We introduce new temporary variables $x_{b-1}^{i}, x_{b-2}^i, \dotsc, x_0^i$, where $\overline{x}^i$ is the value stored in the variable \texttt{x} in the pseudo-code above after the $i$-th iteration. These new variables technically correspond to gates in the circuit; to describe what gates they are we describe the unrolled pseudo-code.
\begin{align*}
\overline{x}^0 := & \overline{x} \\
x_j^i := & \left(\LT_{b-i,p}(\overline{x}^{i-1}) \wedge x_j^{i-1} \right) \vee (\neg \LT_{b-i,p}(\overline{x}^{i-1}) \wedge \SUBPROD_{p,j}(\overline{x}^{i-1}, b-i)) \qquad (i \geq 1) \\
\overline{\REM}_p(\overline{x}) := & \overline{x}^b \\
\end{align*}
Since $\overline{\REM}_p$ has $b$ iterations, and each iteration has size $\poly(b, \log_2 p)$, overall $\REM_p$ has size $\poly(b, \log_2 p)$, as claimed.

We now show that IPS over a field of characteristic $p$ has small (of size $\poly(b, \log_2 p)$) derivations of $\VAL_p(\overline{x}^{i+1}) - \VAL_p(\overline{x}^i)$ for all $i=1,\dotsb,b$. By applying Corollary~\ref{cor:composition}, we then get an IPS derivation of size $\poly(b, \log_2 p)$ of  $\VAL_p(\overline{x}^b) - \VAL_p(\overline{x}^0)$, which is, by definition, the same as $\VAL_p(\overline{\REM}_p(\overline{x})) - \VAL_p(\overline{x})$, as desired.

Fix $i \in \{1,\dots,b\}$. Let 
\begin{align*}
\varphi_j := & \alg(\SUBPROD_{p,j}(\overline{x}^{i-1}, b-i)) \\
\psi := & \alg(\LT_{b-i,p}(\overline{x}^{i-1})
\end{align*}
In the variables $\overline{x}^{i-1}$, we have that  $\VAL_p(\overline{x}^i)$ is
\begin{align*}
\VAL_p(\overline{x}^i) = & \sum_{j=0}^{b-1} 2^j \left( \psi x_j^{i-1} + (1-\psi) \varphi_j + \psi (1-\psi) \varphi_j x_j^{i-1} \right). \\
 = & \psi \sum_{j=0}^{b-1} 2^j x_j^{i-1}  + (1-\psi) \sum_{j=0}^{b-1} 2^j \varphi_j + \psi (1 - \psi) \sum_{j=0}^{b-1} 2^j \varphi_j x_j^{i-1} \\
 = & \psi \VAL_p\left(\overline{\SUBPROD}_{p}(\overline{x}^{i-1}, b-i) \right) \\
  & + (1-\psi) \VAL_p(\overline{x}^{i-1}) + \psi (1-\psi) \chi,
\end{align*}
where $\chi = \sum_j 2^j \varphi_j x_j^{i-1}$.

Since $\chi$ is a linear-sized sum of circuits that are the algebraic translation of Boolean circuits of polynomial size, $\chi$ itself is computed by an algebraic circuit of polynomial size. Since $\psi(1-\psi)$ is the algebraic translation of a polynomial-size Boolean circuit, from the Boolean axioms IPS derives $\psi(1-\psi)\chi$ in polynomial size (e.\,g., \cite[Lem.~3.6]{GrochowPitassi}). This handles the last term, and we are left to handle the first two summands.

By \cite[Lem~5.1]{AGHTarxiv} we have that IPS over $\Z$ derives the following in $\poly(b,\log_2 p)$ size:
\begin{equation} \label{eq:val}
\VAL(\overline{\SUBPROD}_p(\overline{x}^{i-1}, b-i)) - \left(\VAL(\overline{x}^{i-1}) - 2^i p \right)
\end{equation}
Taking that IPS$_\Z$ derivation modulo $p$, we get an IPS$_{\F_p}$ derivation of 
\begin{equation} \label{eq:valp}
\VAL_p(\overline{\SUBPROD}_p(\overline{x}^{i-1}, b-i)) - \VAL_p(\overline{x}^{i-1}),
\end{equation}
since the final term in (\ref{eq:val}) was a multiple of $p$, and is thus zero in $\F_p$.

Now we multiply the preceding derivation of (\ref{eq:valp}) by $\psi$, and we add to it the derivation of $\psi (1-\psi) \chi$ to get an IPS$_{\F_p}$ derivation of size $\poly(b, \log_2 p)$ of 
\[
\psi\left(\VAL_p(\overline{\SUBPROD}_p(\overline{x}^{i-1}, b-i)) - \VAL_p(\overline{x}^{i-1})\right) + \psi (1-\psi) \chi
\]
which is equal, by definition and a tiny algebraic manipulation, to 
\[
\VAL_p(\overline{x}^i) - \psi (\VAL_p(\overline{x}^{i-1}) - (1-\psi)(\VAL_p(\overline{x}^{i-1})).
\]
which then simplifies to
\[
\VAL_p(\overline{x}^i) - \VAL_p(\overline{x}^{i-1}),
\]
as claimed.
\end{proof}

\begin{corollary} \label{cor:addprod}
There are Boolean circuits $\overline{\ADD}_p(\overline{x}, \overline{y})$ and $\overline{\PROD}_p(\overline{x}, \overline{y})$, each with $2b_p$ input bits and $b_p$ output bits, each of size $\poly(b_p)$, such that
\begin{enumerate}
\item $\VAL_+(\overline{\ADD}_p(\overline{x}, \overline{y})) < p$ and $\VAL_+(\overline{\PROD}_p(\overline{x}, \overline{y})) < p$; and

\item IPS over $\F_p$ has proofs of size $\poly(b_p)$ of
\[
\VAL_p(\overline{\ADD}_p(\overline{x}, \overline{y})) - (\VAL_p(\overline{x}) + \VAL_p(\overline{y}))
\]
and
\[
\VAL_p(\overline{\PROD}_p(\overline{x}, \overline{y})) - (\VAL_p(\overline{x}) \times \VAL_p(\overline{y}))
\]
from the Boolean axioms.
\end{enumerate}
\end{corollary}

\begin{proof}
Define $\overline{\ADD}_p(\overline{x}, \overline{y})$ as $\overline{\REM}_p(\overline{\ADD}(\overline{x}, \overline{y}))$ and $\overline{\PROD}_p(\overline{x}, \overline{y})$ as $\overline{\REM}_p(\overline{\PROD}(\overline{x}, \overline{y}))$.

The IPS derivations follow from \cite[Lemma~5.1]{AGHTarxiv} (taking their derivations in IPS$_{\Z}$ modulo $p$), Lemma~\ref{lem:rem}, and one application of the Composition Lemma~\ref{lem:composition}. The size bound follows by observing that in the case of addition, since $\overline{x}$ and $\overline{y}$ are each $b_p$ bits, the output of $\ADD$ is $b_p+1$ bits, so we use $\REM_p$ on $b_p+1$ bits, which has size $\poly(b_p+1)$ by Lemma~\ref{lem:rem}. Similarly, in the case of multiplication, the output of $\PROD$ has at most $2b_p$ bits, so we use $\REM_p$ on $2b_p$ bits, which has size $\poly(2b_p)$.
\end{proof}

We now have all the results and machinery in place to define the mod $p$ analogue of the BIT extraction operator from \cite[Def.~24]{AGHTarxiv}, and to prove the mod $p$ analogue of their binary value lemma.

\begin{definition}[Bit extraction operator modulo $p$, $\overline{\BIT}_p$]
Let $F$ be an algebraic circuit over $\F_p$. For $j=0,\dotsc,b_p-1$, we define $\BIT_{p,j}$ to denote the circuit constructed recursively as follows (which computes the $j$-th bit of the output of $F$). To each algebraic variable $x_i$, we associate $b_p$ Boolean variables $x_{i,0},x_{i,1},\dotsc,x_{i,b_p-1}$.

\begin{itemize}
\item If $F = x_i$ for an (algebraic) variable $x_i$, then $\BIT_{p,j}(F) := x_{i,j}$.

\item If $F = \alpha \in \F_p$, then $\BIT_{p,j}(F)$ is the $j$-th bit of the usual binary representation of $\alpha$, thinking of $\alpha$ as an element of $\{0,\dotsc,p-1\}$.

\item If $F = G + H$, then $\overline{\BIT}_p(F) := \overline{\ADD}_p(\overline{\BIT}_p(G), \overline{\BIT}_p(H))$.

\item If $F = G \times H$, then $\overline{\BIT}_p(F) := \overline{\PROD}_p(\overline{\BIT}_p(G), \overline{\BIT}_p(H))$.

\end{itemize}

\end{definition}

The proof of the following lemma is essentially the same as \cite[Lemma~5.1]{AGHTarxiv}, \emph{mutatis mutandis}, replacing $\BIT$ with $\BIT_p$, $\ADD$ with $\ADD_p$, and $\PROD$ with $\PROD_p$, with an important exception in the base case where $C = x_i$ is a variable, which we discuss in the following remark. 

\begin{remark}
In the case of a variable $x_i$, our $\BIT_p$ operator is \emph{not} merely the mod-$p$ version of the $\BIT$ operator of Alekseev \emph{et al.} \cite{AGHTarxiv}. In their setting, they assume their algebraic variables satisfy the Boolean axioms $x_i^2 - x_i = 0$, whereas in our setting we are attempting to simulate an algebraic circuit not just on Boolean inputs, but on all inputs from a finite field of exponential order. The Boolean assumption in their setting lets them define $\BIT_0(x) = x$ and $\BIT_1(x) = 0$, and then it is a polynomial identity that $\VAL(\overline{\BIT}(x)) = x$, so there is nothing to derive. To achieve the same in our setting, rather than deriving $F - \VAL(\overline{\BIT}(F))$ as they do, we assume that the inputs to $F$ on the left-hand side are already of the form $\VAL_p(\overline{\BIT}_p(x_i))$. Thus, the base case of the following lemma is still a polynomial identity that needs no derivation. Once one has established this base case, and modifying their definition of syntactic length so that the syntactic length of a variable is $b_p$ rather than just 2, the remainder of their proofs are entirely inductive and go through \emph{mutatis mutandis} using our $_p$ operators and the lemmas developed in this section.
\end{remark}

\begin{lemma}[Binary value principle modulo $p$] \label{lem:binvalp}
For any algebraic circuit $F(x_1, \dotsc, x_n)$ of total bit-size $s$ over a prime finite field $\F_p$, there is an IPS$_{\F_p}$ proof of size $\poly(s, \log p)$ of
\[
F(\VAL_p(\overline{\BIT}_p(\overline{x}))) - \VAL_p(\overline{\BIT}_{p}(F)),
\]
from the Boolean axioms for the variables $x_{i,j}$ ($i=1,\dotsc, n, j = 0, \dotsc, b_p-1$).
\end{lemma}

\section{Main Theorem} \label{sec:mainthm}
\begin{theorem} \label{thm:main}
\PIT for circuits of bit-size $n$, over $\Q$ or over finite fields of size $\leq q$, is in $\cc{NP}$ if there is a Cook--Reckhow proof system that p-simulates IPS for Boolean \UNSAT over finite fields of size $\leq \poly(q,2^n)$, with IPS size measured by bit-size.

Conversely, if \PIT[$\F$] is in $\cc{NP}$, then there is a Cook--Reckhow proof system that p-simulates IPS$_\F$ for Boolean \UNSAT, with IPS size measured by bit-size.
\end{theorem}

\begin{proof}
($\Leftarrow$) If PIT is in $\cc{NP}$, then IPS (with size measured as total bit-size) has $\cc{NP}$-verifiable proofs, which is essentially equivalent to being Cook--Reckhow. IPS plus the witness for the $\cc{NP}$ verifier for PIT form a Cook--Reckhow proof system.

($\Rightarrow$) Suppose there is a Cook--Reckhow proof system $\mathcal{P}$ that p-simulates IPS for UNSAT-CNF over fields of order $\leq \poly(q, 2^n)$. We give an $\mathsf{NP}$ algorithm for PIT over $\F_q$ or the rationals. Let $\F \in \{\F_q, \Q\}$. Given an algebraic circuit $C$ of bit-size $n$ over $\F$ consider the following steps. 
\begin{enumerate}
\item By Lemma~\ref{lem:field}, there is a finite field $\K$ of size at most $2^{\poly(n)}$ such that $C$ computes a well-defined function $C_{\K}\colon \K^n \to \K$, and such that $C_{\K}$  is not the zero function if and only if the original circuit $C$ does not compute the identically zero polynomial. Furthermore, $\K$ can be constructed in $\mathsf{TFNP} = \mathsf{NPMV}_{gt}$.

\item From Lemma~\ref{lem:algebraic}, in logspace we then construct a system of equations $\mathcal{F}_C$ over $\K$ such that $C$ is the identically zero polynomial iff $\mathcal{F}_C$ is unsatisfiable over $\K$ iff $\mathcal{F}_C$ is unsatisfiable over $\overline{\K}$. Furthermore, when $\mathcal{F}_C$ is unsatisfiable, there is an IPS$_\K$ proof of this whose size is $\poly(n)$.

\item Let $p = \ch \K$. We extend the $\overline{\BIT}_p$ operator to tuples as $\overline{\BIT}_p((F_1, \dotsc, F_k)) := (\overline{\BIT}_p(F_1), \dotsc, \overline{\BIT}_p(F_k))$. Then
\[
S := \{\overline{\BIT}_p(\overline{\VEC}_{\K/\F_p}(F_i)) : F_i \in \mathcal{F}_C\}
\]
is a set of vectors of Boolean circuits such that $C$ is identically zero as a polynomial iff these Boolean circuits cannot simultaneously all evaluate to zero. In other words, $C \in \PIT$ iff the conjunction of the negations of the Boolean circuits in $S$ is unsatisfiable. 

\item Let $\neg S := \{\neg \Gamma : \Gamma \in S\}$. Reduce each $\neg \Gamma$ from a circuit to CNF as in Section~\ref{sec:prelim:sat}. Since the conjunction of all the $\neg \Gamma$'s ($\Gamma \in S$) is unsatisfiable iff $C \in \PIT$, we may treat the conjunction of all these CNFs as a single, large CNF $\varphi$, and $\varphi$ is unsatisfiable iff $C \in \PIT$. 

\item We claim that, when $\varphi$ is unsatisfiable (equivalently, when $C \in \PIT$), $\varphi$ has a short IPS refutation; we will prove this claim below. Then, by assumption, $\varphi$ thus also has short proofs in the Cook--Reckhow system $\mathcal{P}$. The $\mathsf{NP}$ machine now guesses and verifies a short $\mathcal{P}$-refutation for $\varphi$. If it finds one, it returns YES (the circuit $C$ was identically zero). Otherwise it returns NO.
\end{enumerate}

This completes the description of the $\mathsf{NP}$ algorithm for PIT. All the remains is to prove the claim in the final step, that $\varphi$ has a short IPS refutation.

Suppose $\varphi$ is unsatisfiable. Then IPS refutes $1-\alg(\varphi) = 0$ as follows. By Lemma~\ref{lem:cnf_vs_clauses}, IPS can efficiently derive $1-\alg(\kappa_i)$ for each clause $\kappa_i$ of the CNF $\varphi$. Since $\varphi$ was constructed as the CNF reduction of a conjunction of $\neg \Gamma$ for all $\Gamma \in S$, each $\neg \Gamma$ contributed a certain subset of the clauses $\{\kappa_i\}$. Again by Lemma~\ref{lem:cnf_vs_clauses}, from the individual clauses IPS can efficiently derive the CNF reduct of each $\neg \Gamma$. By Lemma~\ref{lem:circuit-sat}, IPS then efficiently derives $1-\alg(\neg \Gamma)$ for each circuit $\Gamma \in S$. Note that $1-\alg(\neg \Gamma) = \alg(\Gamma)$ (identically as polynomials, by definition of the algebraic translation $\alg(\bullet)$, see Section~\ref{sec:prelim:sat}). Thus, so far, IPS has efficiently derived $\alg(\Gamma)$ for each $\Gamma \in S$.

Now, by definition of $S$, this is the same as having derived $\{\alg(\overline{\BIT}_p(\overline{\VEC}_{\K/\F_p}(F_i))) : F_i \in \mathcal{F}_C\}$, where here we also extend the $\alg$ operator to vectors, namely $\alg((F_1, \dotsc, F_k)) := (\alg(F_1), \dotsc, \alg(F_k))$. We also extend the $\VAL_p$ operator to tuples \emph{mutatis mutandis}. Now, since $\VAL_p(\overline{\BIT}_p(\overline{\VEC}_{\K/\F_p}(F_i)))$ is simply a linear combination of the algebraic circuits in the tuple $\alg(\overline{\BIT}_p(\overline{\VEC}_{\K/\F_p}(F_i)))$, by taking one more linear combination, IPS$_{\F_p}$ has efficiently derived $\VAL_p(\overline{\BIT}_p(\overline{\VEC}_{\K/\F_p}(F_i)))$ for each $F_i \in \mathcal{F}_C$.

Now, separately, for shorthand let $\hat{F}_i := F_i(\VAL_p(\overline{\BIT}_p(\overline{x})))$. By Mod $p$ Binary Value Principle (Lemma~\ref{lem:binvalp}), from the Boolean axioms IPS$_{\F_p}$ efficiently derives $\overline{VEC}_{\K/\F_p}(\hat{F}_i) - \VAL_p(\overline{\BIT}_p(\overline{VEC}_{\K/\F_p}(F_i)))$ for each $F_i \in \mathcal{F}_C$. 

Adding the results of the previous two paragraphs, IPS (so far, just over $\F_p$) has derived $\overline{\VEC}_{\K/\F_p}(\hat{F}_i)$ for each $F_i \in \mathcal{F}_C$. Finally, the operator $\VAL_{\K/\F_p}$ is simply taking a certain $\K$-linear combination of its arguments; by applying this operation IPS$_{\K}$ has now efficiently derived $\VAL_{\K/\F_p}(\overline{\VEC}_{\K/\F_p}(\hat{F}_i))$ for each $F_i \in \mathcal{F}_C$. By Lemma~\ref{lem:extension}, the latter is identically the same as the polynomial $\hat{F}_i$. Thus, so far IPS$_{\K}$ has efficiently derived all of $\mathcal{F}_C \circ (\VAL_p(\overline{\BIT}_p(\overline{x})))$ (that is, each $F_i$ in $\mathcal{F}_C$, composed with $\VAL_p(\overline{\BIT}_p(x_i))$ for each input variable $x_i$).

Finally, we follow the above with a slight twist on the IPS$_\K$ refutation of $\mathcal{F}_C$ from step 2 of the algorithm (which relied on Lemma~\ref{lem:algebraic}), because we don't quite have $\mathcal{F}_C$, but rather we have $\mathcal{F}_C \circ (\VAL_p(\overline{\BIT}_p(\overline{x})))$. However, as IPS derives the constant polynomial $1$ from $\mathcal{F}_C$, we may compose that IPS derivation with $\VAL_p(\overline{\BIT}_p(\overline{x}))$ as well and still get an IPS derivation of 1. That is, if $D(\overline{x}, \overline{y})$ is the IPS certificate refuting $\mathcal{F}_C$, then $D(\VAL_p(\overline{\BIT}_p(\overline{x})), \overline{y})$ is an IPS certificate refuting $\mathcal{F}_C \circ (\VAL_p(\overline{\BIT}_p(\overline{x})))$. This completes the proof that IPS$_\K$ has a short refutation of $\varphi$, thus completing the proof of correctness of the $\mathsf{NP}$ algorithm for PIT.
\end{proof}

\section{Future directions and open questions} \label{sec:conclusion}

\subsection{A route to derandomizing PIT into $\cc{NP}$?} \label{sec:pit}
We now describe in more detail the potential new route to putting \PIT into $\cc{NP}$ that was sketched in \Cref{sec:results}. We begin by recalling the following definition and result from Grochow \& Pitassi \cite{GrochowPitassi}. In the description of their PIT axioms, we follow their notational conventions. Namely, we underline parts that consist of the proposition variables of the relevant Boolean formula. For an algebraic circuit $C(\overline{x})$ in algebraic variables $\overline{x} = x_1, \dotsc, x_n$, we use brackets $[C(\overline{x})]$ to denote the bit-wise description of $C(\overline{x})$. Note that when combined, $\underline{[C(\overline{x})]}$ denotes a collection of Boolean variables which, when assigned values, are interpreted as the description of the algebraic circuit $C(\overline{x})$ (but $C$ is not specified in advance, it depends on the values of the Boolean variables). Tuples of algebraic variables are denoted $\overline{x}, \overline{y}, \dotsc$, while tuples of Boolean variables are denoted $\overline{p}, \overline{q}, \dotsc$. 

\begin{definition}[{PIT axioms, \cite[Def.~5.1]{GrochowPitassi}}]
Let $K$ be a family of Boolean circuits. The PIT axioms for $K$ are:
\begin{enumerate}
\item $K(\underline{[C(\overline{x})]}) \to K(\underline{[C(\overline{p})]}$.

Here, the variables on the left-hand side are Boolean variables $\overline{q}$ encoding an algebraic circuit. On the right-hand side, there are additional Boolean variables $\overline{p}$, and some of the variables of $\overline{q}$---namely, those which describe the input algebraic variables $\overline{x}$---have been replaced by constants or $\overline{p}$ in such a way that $[C(\overline{p})]$ encodes a circuit that plugs in the $\{0,1\}$-valued variables $p_i$ for the input algebraic variables $x_i$.

\item $K(\underline{[C(\overline{x})]}) \to \neg K(\underline{[1-C(\overline{x})]})$.

Here, there is a single set of Boolean variables $\overline{q}$ describing an algebraic circuit $C(\overline{x})$. There is a Boolean function $\varphi$ such that if $\overline{q}$ is the description $[C(\overline{x})]$, then $\varphi(\overline{q})$ is a description of $1-C(\overline{x})$. With this notation, the above axiom is the same as 
\[
K(\overline{q}) \to \neg K(\varphi(\overline{q})).
\]
Similar conventions apply to the remaining axioms.

\item $K(\underline{G(\overline{x})}) \wedge K(\underline{[C(\overline{x}, 0)]}) \to K(\underline{[C(\overline{x}, G(\overline{x}))]})$

\item $K(\underline{[C(\overline{x})]}) \to K(\underline{[C(\pi(\overline{x}))]})$ for all permutations $\pi$ of the $n$ variables $x_1, \dotsc, x_n$.

\end{enumerate}
\end{definition}

Grochow \& Pitassi prove the following result for $\mathcal{C}$-Frege for various circuit classes $\mathcal{C}$, but it is clear that the same proof works to give the following more general statement, \emph{mutatis mutandis}.

\begin{theorem}[{cf. Grochow \& Pitassi \cite[Thms.~1.4 and 1.6]{GrochowPitassi}}] \label{thm:pit}
Let $\mathcal{P}$ be any proof system that implicationally p-simulates $\cc{AC}^0$-Frege. If there is a family $K$ of polynomial-size Boolean circuits solving PIT$_\F$, and such that the PIT axioms for $K$ have polynomial-size $\mathcal{P}$-proofs, then $\mathcal{P}$ p-simulates IPS$_{\F}$ for UNSAT-CNF (with size in IPS measured by total bit-size).
\end{theorem}

Our suggestion of how one might prove that PIT is in $\cc{NP}$ is then encapsulated in the following corollary:

\begin{corollary} \label{cor:pit}
Let $\mathcal{P}$ be any Cook--Reckhow proof system that implicationally p-simulates $\cc{AC}^0$-Frege. If there is a family $K$ of polynomial-size Boolean circuits that correctly solves PIT$_\F$ over fields of size $\leq \poly(q, 2^{n})$, and such that the PIT axioms for $K$ have polynomial-size $\mathcal{P}$-proofs, then PIT---for circuits over $\Q$ or finite fields of size $\leq q$---is in $\cc{NP}$.
\end{corollary}

\begin{proof}
Under these hypotheses, by Theorem~\ref{thm:pit}, IPS$_{\F}$ over fields of size $\leq \poly(q, 2^{n})$ is p-simulated by the Cook--Reckhow system $\mathcal{P}$. By our Main Theorem~\ref{thm:main}, it follows that PIT over $\Q$ or fields of size $\leq q$ is in $\cc{NP}$.
\end{proof}

One aspect we find potentially interesting about this approach is that, on the one hand, since PIT is in $\cc{BPP}$, we know that PIT is in $\cc{P/poly}$, so in some sense the ``whole question'' of derandomizing PIT is one of ``\emph{uniformizing}'' PIT, in the sense of ``removing the use of non-uniformity.'' On the other hand, in Corollary~\ref{cor:pit}, neither the circuit family $K$ nor the short $\mathcal{P}$-proofs for the PIT axioms need to be uniform; they can be non-uniform, and as long as they \emph{exist} (and have polynomial size), it implies the \emph{uniform} conclusion that PIT is in $\cc{NP}$ (essentially, via the uniformity in the proof-checker for $\mathcal{P}$). This particular aspect of this proposal makes this approach feel, at least to this author, like it would be significantly different from approaches based on pseudo-random generators or on unconditional derandomization of PIT for specific circuit classes, either of which seem to require more uniform solutions.

\subsection{Open questions} \label{sec:open}
\begin{open}
Can we improve \Cref{thm:main} for IPS for \UNSAT to use the same field for both IPS and \PIT, as the result for \lang{Variety Emptiness} (\Cref{thm:mainpoly}) does?
\end{open}
In either the case of finite fields or the rationals, this would seem to need a new approach; see \Cref{rmk:field} for more discussion.

\begin{open}
What is the relationship between IPS over different fields, especially when viewed as proof systems for Boolean \UNSAT? 
\end{open}
For two distinct primes $p$ and $q$, if we take an unsatisfiable system of polynomial equations over $\F_p$, and use the techniques of this paper to produce an unsatisfiable CNF $\varphi$ such that IPS$_{\F_p}$ can derive the original polynomial equations from $\varphi$, is $\varphi$ hard for IPS$_{\F_q}$? 

We note that even for extension fields the answer is not immediately apparent. With the Nullstellensatz or PC proof systems, a certificate exists over an extension field $\K \supseteq \F$ if and only if a certificate exists over the ground field $\F$; this follows because certificates in those systems can be viewed as solutions to certain (unions of exponentially large) linear equations over $\F$, and linear equations have the property that they have solutions over an extension field iff they have solutions over the ground field. Interestingly, because of the equivalence between linear $\cc{VP}_{det}$-IPS and PC (where proof size in PC is measured by number of lines) \cite[Prop.~3.4]{GrochowPitassi}, this tells us that the power of linear $\cc{VP}_{det}$-IPS only depends at most on the characteristic of the field. But for general IPS we have no such equivalence (though IPS certificates are the solutions of \emph{polynomial} equations, see the proof of \cite[Prop.~3.2]{GrochowPitassi}, with Koiran). 

We may similarly ask about \PIT over different fields:

\begin{open}
What is the relationship between the various versions of \PIT over different fields?
\end{open}

\textit{Other fields.} Lastly, can our results be extended to other fields, such as algebraic number fields or $\mathbb{C}$? Over algebraic number fields we suspect the answer is yes, using just a little additional number theory to get an analogue of \Cref{lem:field}. Over $\mathbb{C}$ we also suspect the answer is yes, using methods similar to Koiran \cite{koiran}. Over something like a function field, or the field of fractions of the coordinate ring of a variety, we do not have strong intuition about the result, but expect it to be quite a bit more complicated to resolve.

\section*{Origin of the paper\footnote{This is something I'm trying, with the thought that sections like this could be useful for aspiring researchers in the future, who wonder ``How did they even think to work on this in the first place? How did they decide to?'' I hope others will join me in this experiment in pulling back the curtain.}}
\label{sec:origin}
Ever since we worked on IPS in 2013, it was a natural question as to whether IPS verification could be done deterministically, somehow avoiding the worst case of \PIT, and simultaneously putting IPS into the class of Cook--Reckhow proof systems. After \Hrubes and an anonymous reviewer pointed out to us that it was not even obvious that ZFC p-simulated IPS (see p. 4, \Cref{sec:intro}), Toni Pitassi and I talked about what difficulty ZFC might have in p-simulating IPS. A (seemingly) key issue was that it was unclear whether ZFC---as a proof system for Boolean \UNSAT---could prove the PIT Lemma (reproduced as \Cref{lem:pit} above), whose proof is based on a (probabilistic) counting argument. If this was not obvious even for such a powerful proof system as ZFC, it was natural to wonder whether the same difficulty would be encountered by trying to p-simulate IPS by \emph{any one} Cook--Reckhow proof system, and hence whether such a p-simulation entailed some derandomization of \PIT.

In February of 2022, while unpacking some boxes, I was trying to think of new projects to work on with Toni Pitassi, and decided (somewhat randomly) to think about this question again. How could one encode the identically vanishing of an algebraic circuit $C$ into an unsatisfiable system of polynomial equations? The natural, well-known trick from algebraic geometry (going back probably to Hilbert if not earlier) is to add an equation like $1-zC = 0$, which forces the output of $C$ to be invertible. And the natural trick from circuit complexity, going back at least to Ben-Or \cite{BenOr} if not earlier, is to add a new variable for each gate, and equations enforcing that the gate variables compute the polynomials at the gate. This combination led to \Cref{lem:algebraic}, which is what made it clear the probably the rest could be worked out. I didn't realize how long it would take and how much work it would be! 

\section*{Acknowledgments}

I would like to thank Toni Pitassi for helpful discussions in the early stages of this work, and Pavel \Hrubes for pointing out the initial error about ZFC that eventually led to the question resolved in this paper. I would also like to thank A. Atserias, J. Nordström, P. Pudlák, and R. Santhanam for organizing and inviting me to the Dagstuhl Seminar 18051: Proof Complexity, in January 2018, where the conversation with \Hrubes occurred; and E. Allender, A. Kolokolova, P. Papakonstantinou, and R. Santhanam for organizing, inviting me to present at, and accommodating my need for remote presentation at the DIMACS Workshop on Meta-Complexity, Barriers, and Derandomization, where I presented on a preliminary version of this work (recording available \href{https://youtu.be/7bzd9DaMHRQ}{here}) and had many interesting discussions.  This work was supported by NSF CAREER award CISE-2047756. 

\bibliographystyle{alphaurl}
\bibliography{ips-pit}

  \newcommand{\conference}[4]{{#1} '#3: #4
  #2}\newcommand{\FOCS}[2]{\conference{FOCS}{Annual {IEEE} {Symposium} on
  {Foundations} of {Computer}
  {Science}}{#1}{#2}}\newcommand{\STOC}[2]{\conference{STOC}{Annual {ACM}
  {Symposium} on {Theory} of
  {Computing}}{#1}{#2}}\newcommand{\MFCS}[2]{\conference{MFCS}{{Symposium} on
  {Mathematical} {Foundations} of {Computer}
  {Science}}{#1}{#2}}\newcommand{\SODA}[2]{\conference{SODA}{{ACM}--{SIAM}
  {Symposium} on {Discrete}
  {Algorithms}}{#1}{#2}}\newcommand{\ICALP}[2]{\conference{ICALP}{{International}
  {Colloquium} on {Automata}, {Languages} and
  {Programming}}{#1}{#2}}\newcommand{\CCC}[2]{\conference{CCC}{{IEEE}
  {Conference} on {Computational}
  {Complexity}}{#1}{#2}}\newcommand{\STACS}[2]{\conference{STACS}{Annual
  {Symposium} on {Theoretical} {Aspects} of {Computer}
  {Science}}{#1}{#2}}\newcommand{\ISSAC}[2]{\conference{ISSAC}{{International}
  {Symposium} on {Symbolic} and {Algebraic}
  {Computation}}{#1}{#2}}\newcommand{\original}[1]{A preliminary version
  appeared in {#1}}\newcommand{\preprint}[1]{preprint
  {#1}}\newcommand{\available}[1]{also available as \cite{#1}}
\begin{thebibliography}{AGHT20}

\bibitem[AF22]{AndrewsForbes}
Robert Andrews and Michael~A. Forbes.
\newblock Ideals, determinants, and straightening: proving and using lower
  bounds for polynomial ideals.
\newblock In Stefano Leonardi and Anupam Gupta, editors, {\em \STOC{22}{54th}},
  pages 389--402. {ACM}, 2022.
\newblock \href {https://doi.org/10.1145/3519935.3520025}
  {\path{doi:10.1145/3519935.3520025}}.

\bibitem[AGHT19]{AGHTarxiv}
Yaroslav Alekseev, Dima Grigoriev, Edward~A. Hirsch, and Iddo Tzameret.
\newblock Semi-algebraic proofs, {IPS} lower bounds and the
  {\(\tau\)}-conjecture: Can a natural number be negative?
\newblock \href{http://arxiv.org/abs/1911.06738}{arXiv:1911.06738 [cs.CC]},
  2019.
\newblock Preprint of full version of \cite{AGHT}.

\bibitem[AGHT20]{AGHT}
Yaroslav Alekseev, Dima Grigoriev, Edward~A. Hirsch, and Iddo Tzameret.
\newblock Semi-algebraic proofs, {IPS} lower bounds, and the
  {\(\tau\)}-conjecture: can a natural number be negative?
\newblock In Konstantin Makarychev, Yury Makarychev, Madhur Tulsiani, Gautam
  Kamath, and Julia Chuzhoy, editors, {\em \STOC{20}{52nd}}, pages 54--67.
  {ACM}, 2020.
\newblock \href {https://doi.org/10.1145/3357713.3384245}
  {\path{doi:10.1145/3357713.3384245}}.

\bibitem[AKS04]{AKS}
Manindra Agrawal, Neeraj Kayal, and Nitin Saxena.
\newblock {PRIMES} is in {P}.
\newblock {\em Ann. of Math. (2)}, 160(2):781--793, 2004.
\newblock \href {https://doi.org/10.4007/annals.2004.160.781}
  {\path{doi:10.4007/annals.2004.160.781}}.

\bibitem[ALJ86]{AL86}
Leonard~M. Adleman and Hendrik~W. Lenstra~Jr.
\newblock Finding irreducible polynomials over finite fields.
\newblock In Juris Hartmanis, editor, {\em Proceedings of the 18th Annual {ACM}
  Symposium on Theory of Computing, May 28-30, 1986, Berkeley, California,
  {USA}}, pages 350--355. {ACM}, 1986.
\newblock \href {https://doi.org/10.1145/12130.12166}
  {\path{doi:10.1145/12130.12166}}.

\bibitem[Ben83]{BenOr}
Michael Ben{-}Or.
\newblock Lower bounds for algebraic computation trees (preliminary report).
\newblock In David~S. Johnson, Ronald Fagin, Michael~L. Fredman, David Harel,
  Richard~M. Karp, Nancy~A. Lynch, Christos~H. Papadimitriou, Ronald~L. Rivest,
  Walter~L. Ruzzo, and Joel~I. Seiferas, editors, {\em \STOC{83}{15th}}, pages
  80--86. {ACM}, 1983.
\newblock \href {https://doi.org/10.1145/800061.808735}
  {\path{doi:10.1145/800061.808735}}.

\bibitem[BSS89]{BSS}
Lenore Blum, Mike Shub, and Steve Smale.
\newblock On a theory of computation and complexity over the real numbers:
  {NP}-completeness, recursive functions and universal machines.
\newblock {\em Bull. Amer. Math. Soc. (N.S.)}, 21(1):1--46, 1989.
\newblock \href {https://doi.org/10.1090/S0273-0979-1989-15750-9}
  {\path{doi:10.1090/S0273-0979-1989-15750-9}}.

\bibitem[CCT87]{CCT}
W.~Cook, C.~R. Coullard, and Gy. Tur\'{a}n.
\newblock On the complexity of cutting-plane proofs.
\newblock {\em Discrete Appl. Math.}, 18(1):25--38, 1987.
\newblock \href {https://doi.org/10.1016/0166-218X(87)90039-4}
  {\path{doi:10.1016/0166-218X(87)90039-4}}.

\bibitem[CEI96]{CEI}
Matthew Clegg, Jeffery Edmonds, and Russell Impagliazzo.
\newblock Using the {Groebner} basis algorithm to find proofs of
  unsatisfiability.
\newblock In {\em \STOC{96}{28th}}, pages 174--183. ACM, New York, 1996.
\newblock \href {https://doi.org/10.1145/237814.237860}
  {\path{doi:10.1145/237814.237860}}.

\bibitem[Chv73]{chvatal1973}
V.~Chv\'{a}tal.
\newblock Edmonds polytopes and a hierarchy of combinatorial problems.
\newblock {\em Discrete Math.}, 4:305--337, 1973.
\newblock \href {https://doi.org/10.1016/0012-365X(73)90167-2}
  {\path{doi:10.1016/0012-365X(73)90167-2}}.

\bibitem[CIKK15]{CIKK}
Marco Carmosino, Russell Impagliazzo, Valentine Kabanets, and Antonina
  Kolokolova.
\newblock Tighter connections between derandomization and circuit lower bounds.
\newblock In Naveen Garg, Klaus Jansen, Anup Rao, and Jos{\'{e}} D.~P. Rolim,
  editors, {\em Approximation, Randomization, and Combinatorial Optimization.
  Algorithms and Techniques, {APPROX/RANDOM} 2015, August 24-26, 2015,
  Princeton, NJ, {USA}}, volume~40 of {\em LIPIcs}, pages 645--658. Schloss
  Dagstuhl - Leibniz-Zentrum f{\"{u}}r Informatik, 2015.
\newblock \href {https://doi.org/10.4230/LIPIcs.APPROX-RANDOM.2015.645}
  {\path{doi:10.4230/LIPIcs.APPROX-RANDOM.2015.645}}.

\bibitem[CR79]{CookReckhow}
Stephen~A. Cook and Robert~A. Reckhow.
\newblock The relative efficiency of propositional proof systems.
\newblock {\em J. Symb. Log.}, 44(1):36--50, 1979.
\newblock Some results here appeared in preliminary form in STOC '74 and
  Reckhow's Ph.D. thesis (U. Toronto Dept. of Comp. Sci., 1976).
\newblock \href {https://doi.org/10.2307/2273702} {\path{doi:10.2307/2273702}}.

\bibitem[DL78]{DML}
Richard~A. DeMillo and Richard~J. Lipton.
\newblock A probabilistic remark on algebraic program testing.
\newblock {\em Inf. Process. Lett.}, 7(4):193--195, 1978.
\newblock \href {https://doi.org/10.1016/0020-0190(78)90067-4}
  {\path{doi:10.1016/0020-0190(78)90067-4}}.

\bibitem[FSTW21]{FSTW}
Michael~A. Forbes, Amir Shpilka, Iddo Tzameret, and Avi Wigderson.
\newblock Proof complexity lower bounds from algebraic circuit complexity.
\newblock {\em Theory Comput.}, 17:1--88, 2021.
\newblock \original{CCC '16}.
\newblock \href {https://doi.org/10.4086/toc.2021.v017a010}
  {\path{doi:10.4086/toc.2021.v017a010}}.

\bibitem[GHT22]{GHT}
Nashlen Govindasamy, Tuomas Hakoniemi, and Iddo Tzameret.
\newblock Simple hard instances for low-depth algebraic proofs.
\newblock In {\em \FOCS{22}{63rd}}, pages 188--199. {IEEE}, 2022.
\newblock \href {https://doi.org/10.1109/FOCS54457.2022.00025}
  {\path{doi:10.1109/FOCS54457.2022.00025}}.

\bibitem[Gom63]{gomory1963}
Ralph~E. Gomory.
\newblock An algorithm for integer solutions to linear programs.
\newblock In {\em Recent advances in mathematical programming}, pages 269--302.
  McGraw-Hill, New York, 1963.

\bibitem[GP18]{GrochowPitassi}
Joshua~A. Grochow and Toniann Pitassi.
\newblock Circuit complexity, proof complexity, and polynomial identity
  testing: The ideal proof system.
\newblock {\em J. ACM}, 65:37, 2018.
\newblock Preliminary version appeared in FOCS 2014 (doi:10.1109/FOCS.2014.20).
\newblock \href {https://doi.org/10.1145/3230742} {\path{doi:10.1145/3230742}}.

\bibitem[Gro19]{certppoly}
Joshua~A. Grochow.
\newblock Answer to ``what would signify hierarchy collapse to first level?''
  on \url{cstheory.stackexchange.com}.
\newblock \url{https://cstheory.stackexchange.com/a/45760/129}, 2019.

\bibitem[GV01]{GV}
Dima Grigoriev and Nicolai Vorobjov.
\newblock Complexity of {Null-} and {Positivstellensatz} proofs.
\newblock {\em Annals of Pure and Applied Logic}, 113(1):153--160, 2001.
\newblock First St. Petersburg Conference on Days of Logic and Computability.
\newblock \href {https://doi.org/10.1016/S0168-0072(01)00055-0}
  {\path{doi:10.1016/S0168-0072(01)00055-0}}.

\bibitem[Ier89]{ierardi}
D.~Ierardi.
\newblock Quantifier elimination in the theory of an algebraically-closed
  field.
\newblock In {\em \STOC{89}{21st}}, pages 138--147, New York, NY, USA, 1989.
  Association for Computing Machinery.
\newblock \href {https://doi.org/10.1145/73007.73020}
  {\path{doi:10.1145/73007.73020}}.

\bibitem[JS12]{JansenSanthanam}
Maurice~J. Jansen and Rahul Santhanam.
\newblock Stronger lower bounds and randomness-hardness trade-offs using
  associated algebraic complexity classes.
\newblock In Christoph D{\"{u}}rr and Thomas Wilke, editors, {\em
  \STACS{12}{29th}}, volume~14 of {\em LIPIcs}, pages 519--530. Schloss
  Dagstuhl - Leibniz-Zentrum f{\"{u}}r Informatik, 2012.
\newblock \href {https://doi.org/10.4230/LIPIcs.STACS.2012.519}
  {\path{doi:10.4230/LIPIcs.STACS.2012.519}}.

\bibitem[KI04]{KI}
Valentine Kabanets and Russell Impagliazzo.
\newblock Derandomizing polynomial identity tests means proving circuit lower
  bounds.
\newblock {\em Comput. Complexity}, 13(1-2):1--46, 2004.
\newblock \href {https://doi.org/10.1007/s00037-004-0182-6}
  {\path{doi:10.1007/s00037-004-0182-6}}.

\bibitem[Koi96]{koiran}
Pascal Koiran.
\newblock Hilbert's {Nullstellensatz} is in the polynomial hierarchy.
\newblock {\em J. Complexity}, 12(4):273--286, 1996.
\newblock Special issue for the Foundations of Computational Mathematics
  Conference (Rio de Janeiro, 1997).
\newblock \href {https://doi.org/10.1006/jcom.1996.0019}
  {\path{doi:10.1006/jcom.1996.0019}}.

\bibitem[LTW18]{LTW}
Fu~Li, Iddo Tzameret, and Zhengyu Wang.
\newblock Characterizing propositional proofs as noncommutative formulas.
\newblock {\em {SIAM} J. Comput.}, 47(4):1424--1462, 2018.
\newblock Originally appeared in CCC '15 (doi:10.4230/LIPIcs.CCC.2015.412).
\newblock \href {https://doi.org/10.1137/16M1107632}
  {\path{doi:10.1137/16M1107632}}.

\bibitem[Mos10]{moshkovitz}
Dana Moshkovitz.
\newblock An alternative proof of the {Schwartz}--{Zippel} {Lemma}.
\newblock \href{https://eccc.weizmann.ac.il//report/2010/096/}{ECCC Tech.
  Report TR10-096}, 2010.

\bibitem[MP91]{MP}
Nimrod Megiddo and Christos~H. Papadimitriou.
\newblock On total functions, existence theorems and computational complexity.
\newblock {\em Theor. Comput. Sci.}, 81(2):317--324, 1991.
\newblock \href {https://doi.org/10.1016/0304-3975(91)90200-L}
  {\path{doi:10.1016/0304-3975(91)90200-L}}.

\bibitem[Ore22]{ore}
Øystein Ore.
\newblock Über höhere {Kongruenzen}.
\newblock {\em Norsk Mat. Forenings Skrifter, Ser. I}, (7), 1922.

\bibitem[Pit96]{Pitassi96}
Toniann Pitassi.
\newblock Algebraic propositional proof systems.
\newblock In {\em Descriptive Complexity and Finite Models, Proceedings of the
  DIMACS Workshop held at Princeton University, Princeton, NJ, January
  14–`17, 1996. Edited by Neil Immerman and Phokion G. Kolaitis}, volume~31
  of {\em DIMACS Series in Discrete Mathematics and Theoretical Computer
  Science}, pages 215--244. American Mathematical Society, 1996.

\bibitem[Pit98]{Pitassi98}
Toniann Pitassi.
\newblock Propositional proof complexity and unsolvability of polynomial
  equations.
\newblock In {\em Proceedings of the {International} {Congress} of
  {Mathematicians}. Vol. III. Sections 10--19. Held in Berlin, August 18-–27,
  1998}, pages 215--244, 1998.

\bibitem[Raz87]{Razborov}
Alexander~A. Razborov.
\newblock Lower bounds on the dimension of schemes of bounded depth in a
  complete basis containing the logical addition function.
\newblock {\em Mat. Zametki}, 41(4):598--607, 623, 1987.
\newblock {English} translation: Mathematical Notes of the Academy of Sci. of
  the USSR, 41(4):333--338, 1987.

\bibitem[Rob65]{robinson}
J.~A. Robinson.
\newblock A machine-oriented logic based on the resolution principle.
\newblock {\em J. ACM}, 12(1):23--41, 1965.
\newblock \href {https://doi.org/10.1145/321250.321253}
  {\path{doi:10.1145/321250.321253}}.

\bibitem[RS05]{RazShpilka}
Ran Raz and Amir Shpilka.
\newblock Deterministic polynomial identity testing in non-commutative models.
\newblock {\em Comput. Complex.}, 14(1):1--19, 2005.
\newblock Originally appeared in CCC '04 (doi:10.1109/CCC.2004.1313845).
\newblock \href {https://doi.org/10.1007/s00037-005-0188-8}
  {\path{doi:10.1007/s00037-005-0188-8}}.

\bibitem[Sch80]{Schwartz}
J.~T. Schwartz.
\newblock Fast probabilistic algorithms for verification of polynomial
  identities.
\newblock {\em J. ACM}, 27(4):701–717, oct 1980.
\newblock \href {https://doi.org/10.1145/322217.322225}
  {\path{doi:10.1145/322217.322225}}.

\bibitem[Sel94]{selman}
Alan~L. Selman.
\newblock A taxonomy of complexity classes of functions.
\newblock {\em J. Comput. Syst. Sci.}, 48(2):357--381, 1994.
\newblock \href {https://doi.org/10.1016/S0022-0000(05)80009-1}
  {\path{doi:10.1016/S0022-0000(05)80009-1}}.

\bibitem[Sho90]{shoup}
Victor Shoup.
\newblock New algorithms for finding irreducible polynomials over finite
  fields.
\newblock {\em Math. Comp.}, 54(189):435--447, 1990.
\newblock \href {https://doi.org/10.2307/2008704} {\path{doi:10.2307/2008704}}.

\bibitem[Smo87]{Smolensky}
Roman Smolensky.
\newblock Algebraic methods in the theory of lower bounds for {B}oolean circuit
  complexity.
\newblock In {\em \STOC{87}{19th}}, pages 77--82. ACM, 1987.
\newblock \href {https://doi.org/10.1145/28395.28404}
  {\path{doi:10.1145/28395.28404}}.

\bibitem[ST21]{SanthanamTzameret}
Rahul Santhanam and Iddo Tzameret.
\newblock Iterated lower bound formulas: a diagonalization-based approach to
  proof complexity.
\newblock In Samir Khuller and Virginia~Vassilevska Williams, editors, {\em
  \STOC{21}{53rd}}, pages 234--247. {ACM}, 2021.
\newblock Preliminary full version available as ECCC Tech. Report TR21-138.
\newblock \href {https://doi.org/10.1145/3406325.3451010}
  {\path{doi:10.1145/3406325.3451010}}.

\bibitem[Val79]{valiant}
Leslie~G. Valiant.
\newblock Completeness classes in algebra.
\newblock In {\em \STOC{79}{11th}}, pages 249--261. ACM, 1979.
\newblock \href {https://doi.org/10.1145/800135.804419}
  {\path{doi:10.1145/800135.804419}}.

\bibitem[vzGG13]{GGbook}
Joachim von~zur Gathen and J\"{u}rgen Gerhard.
\newblock {\em Modern computer algebra}.
\newblock Cambridge University Press, Cambridge, third edition, 2013.
\newblock \href {https://doi.org/10.1017/CBO9781139856065}
  {\path{doi:10.1017/CBO9781139856065}}.

\bibitem[Zip79]{Zippel}
Richard Zippel.
\newblock Probabilistic algorithms for sparse polynomials.
\newblock In Edward~W. Ng, editor, {\em Symbolic and Algebraic Computation,
  {EUROSAM} '79, An International Symposiumon Symbolic and Algebraic
  Computation, Marseille, France, June 1979, Proceedings}, volume~72 of {\em
  Lecture Notes in Computer Science}, pages 216--226. Springer, 1979.
\newblock \href {https://doi.org/10.1007/3-540-09519-5\_73}
  {\path{doi:10.1007/3-540-09519-5\_73}}.

\end{thebibliography}

\end{document}